\documentclass[12pt]{article}

\usepackage{soul}
\usepackage{graphicx,color}  % needed for figures
\usepackage{amsfonts,amssymb,amsthm}
\usepackage{amsmath}
\usepackage{mathtools}
\usepackage{tabularx}
\usepackage{comment}
\usepackage{cancel}
\usepackage{bigints}
\usepackage{soul}
\usepackage{ulem}
\usepackage{scalerel,stackengine}
\usepackage{float}

\usepackage{enumitem}

\stackMath
\theoremstyle{definition}
\newcommand\reallywidehat[1]{%
\savestack{\tmpbox}{\stretchto{%
  \scaleto{%
    \scalerel*[\widthof{\ensuremath{#1}}]{\kern.1pt\mathchar"0362\kern.1pt}%
    {\rule{0ex}{\textheight}}%WIDTH-LIMITED CIRCUMFLEX
  }{\textheight}% 
}{2.4ex}}%
\stackon[-6.9pt]{#1}{\tmpbox}%
}
\newtheorem{theorem}{Theorem}
\newtheorem{corollary}{Corollary}[theorem]
\newtheorem{lemma}{Lemma}[theorem]
\newtheorem{remark}{Remark}[theorem]
\newtheorem{example}{Example}[theorem]
\newtheorem{proposition}{Proposition}[theorem]
\usepackage[colorlinks,citecolor=red,urlcolor=blue,bookmarks=false,hypertexnames=true]{hyperref}

\usepackage{titling}

\begin{document}

% Title and authors (short title not supported directly in 'article')
\title{Three-level qualitative classification of financial risks under varying conditions through first passage times}

%Alternative title: 

\setlength{\droptitle}{-2cm}  % Reduce space above the title
%\postdate{\vspace{-0.5 cm}}
%\title{First passage times of stochastic processes through arbitrary barriers: Part I. \\{\large Almost sure and $L^1$ finiteness and three-level qualitative classification of credit risks in dynamical scenarios.}}
\author{
  Carlos Bouthelier-Madre\thanks{carlosbouthelier@gmail.com},
  Carlos Escudero\thanks{cescudero@mat.uned.es} \\
  \small Departamento de Matem\'aticas Fundamentales,\\ \small Universidad Nacional de Educaci\'on a Distancia, Madrid, Spain
}
\date{\today} % Or specify manually
\maketitle
\vspace{-0.5 cm}
\begin{abstract}
This work focuses on financial risks from a probabilistic point of view. The value of a firm is described as a geometric Brownian motion and default emerges as a first passage time event. On the technical side, the critical threshold that the value process has to cross to trigger the default is assumed to be an arbitrary continuous function, what constitutes a generalization of the classical Black-Cox model. Such a generality favors modeling a wide range of risk scenarios, including those characterized by strongly time-varying conditions; but at the same time limits the possibility of obtaining closed-form formulae. To avoid this limitation, we implement a qualitative classification of risk into three categories: high, medium, and low. They correspond, respectively, to a finite mean first passage time, to an almost surely finite first passage time with infinite mean, and to a positive probability of survival for all times. This allows for an extensive classification of risk based only on the asymptotic behavior of the default function, which generalizes previously known results that assumed this function to be an exponential. However, even within these mathematical conditions, such a classification is not exhaustive, as a consequence of the behavioral freedom that continuous functions enjoy. Overall, our results contribute to the design of credit risk classifications from analytical principles and, at the same time, constitute a call of attention on potential models of risk assessment in situations largely affected by time evolution.
\end{abstract}

\\ \indent 2020 \textit{MSC: 60G40, 60J70, 91G40.}

\section{Introduction}

Credit risk assessment has been recognized for decades as a topic of key importance in the field of mathematical finance. Evaluating the risk of default is far from being a mere motivation for theoretical questions, but an issue of great practical application~\cite{credit2004bielecki}. Default events are often modeled as first passage time (FPT) events: they occur at the first time a company asset value crosses a critical threshold, below which liabilities can no longer be met. This has a natural formulation within probability theory as a FPT problem. In fact, that is the formulation chosen by the classical Black-Cox model of default risk~\cite{black1976valuing,brigo2009credit}.

In this work, we incorporate the main assumptions of the Black-Cox model: the firm asset value will be assumed to be a geometric Brownian motion and the critical threshold signaling default will be a continuous function of time, from now on called ``the barrier function''. Within this field,  there is an abundant literature focused on finding closed-form expressions for the mean of the FPT or the probability of default. Therefore, the classical Black-Cox framework concentrates on a limited range of barriers, such as exponential or constant functions, that allows the derivation of explicit formulae; some comprehensive treatments can be found, for example, in~\cite{bielecki2001credit,jeanblanc2009mathematical,karatzas1998methods}. These results were foundational and paved the way for the application of stochastic analysis in structural credit risk theory.

While these simple barrier shapes have their range of application, there is an undeniable interest in the study of first passage problems through more complex barriers. However, more involved barrier shapes tend to be elusive to its explicit resolution. In fact, the diversity of approaches employed to tackle first passage problems is vast and goes well beyond the realm of stochastic analysis to include partial differential equation methods~\cite{fox1986functional,tuckwell1984first} and integral equation approaches, such as Volterra or Fredholm equations~\cite{Jaimungal,peskir2006optimal}; but still they produce a limited variety of explicitly solvable examples. This has led to the use of piecewise definitions for the barriers~\cite{fu2010linear,guillaume2024computation,jin2017first}, where either the pieces are taken to be elementary functions that are solvable, or their validity needs to be falsified against the numerical resolution of the problem~\cite{fu2010linear}.

Yet, a different attempt of generalization of the classical theory of structural credit risk is the one pushed forward by the so-called Analytically Tractable FPT (AT1P) models, see \cite{brigo2009credit,molini2011first,rapisarda2005pricing}. These models allow to find more solutions at the price of introducing an interdependence between barrier shape and volatility. In particular, only barrier functions, henceforth denoted $B(t)$, of the precise form
$$
\ln[B(t)/B(0)] \propto \int\limits_0^t\sigma^2(s) \, ds,
$$
are permitted, where $\sigma(t)$ denotes the time-dependent volatility and $t \ge 0$. Indeed, AT1P models highlight, on one hand, the interest in obtaining new explicitly solvable examples and, on the other, the difficulty in getting so. This interest is not only theoretically motivated, but has a practical aspect connected to the intensity-based models used in real markets: the deduction of closed-form expressions for FPTs in structural models with more general barriers could help refine real-world market-calibrated models~\cite{duffie2001term,jarrow2012structural}.

Given this clear trend in obtaining full mathematical results for FPTs with arbitrary barriers, which is hindered by the troublesome search of new explicitly solvable examples, we will try herein a different approach. We propose a qualitative classification of risk into three levels: low, medium, and high. The low risk is identified with a positive probability of not entering into default at any time. High risk is associated with a finite mean FPT, which provides a characteristic time scale for the default to happen.  The intermediate situation is characterized by an almost surely finite first passage time with divergent first moment. Such a classification is natural from the probabilistic viewpoint, as it does not rely on arbitrary numerical cut-offs, and opens the possibility for proving fully rigorous results. Part of our motivation also comes from the physics of phase transitions~\cite{quantum2021zinn}, which has already inspired other classifications of stochastic processes~\cite{noise1984horsthemke}, as in our case a risk status can change abruptly as the external conditions are modified. That is why we refer to {\it criticality} as the set of conditions that separates two risk statuses. 

A big part of our efforts focuses on near-exponentially moving barrier functions. The reason behind this choice is the fact that classical results on the Black-Cox model already successfully classify functions that move much faster or much slower than exponentials. Therefore, the outcome of our analysis is a finer classification of risk that is able to sort barriers whose behavior is doubtful or undecidable from the viewpoint of the standard Black-Cox paradigm. To achieve this goal we will make use of the properties of Brownian motion, proving results that are instrumental to our objectives. The mathematical machinery we will employ is that established in the field and recorded in monographs such as~\cite{handbook2012borodin,karatzas1991brownian,brownian2010morters,brownian2014sp}.

The outline of the paper is as follows. In  Section~\ref{sec:Inflation} we state the problem and fix our notation. We also emphasize the relevance of considering arbitrary barriers in finance, illustrating it with the effect of economic inflation on credit risk assessment, although other effects, such as that of extreme weather, are very well possible. In Section~\ref{sec:main_FPT}, we study the finer time dependencies of modeling safety covenants for future debt payment when a company is undecidable under the classical Black-Cox structural model of credit risk. We do it mathematically by analyzing FPTs for barriers that are neither constant nor exponential functions, and using our classification system. In Section~\ref{sec:critical_prob} we refine our analysis with regard to the separation between the low and intermediate risk zones. In Section~\ref{sec:mfpts} we deepen the analysis with respect to the separation between the intermediate and high risk zones. Some examples of how our theory can be improved in particular cases are set forth in Section~\ref{sec:examples}. Finally, in Section~\ref{sec:conclusions} we draw our main conclusions.

\section{The problem of FPTs for moving barriers and the role of inflation}
\label{sec:Inflation}

In this section, we state the general problem of first passage times through moving barriers as we will treat it to match structural models of default risk. We depart from the well-known Black-Cox model. In this model, a stochastic process $V_t$ solves the Black-Scholes equation:
$$
dV_t=\mu V_t dt + \sigma V_t dW_t, \;\; \left. V_t \right|_{t=0}= V_0\quad\Rightarrow\quad V_t=V_0\exp((\mu-\sigma^2/2)t+\sigma W_t),
$$
with $\mu \in \mathbb{R}$ and $\sigma, V_0$ being positive constants. The default event occurs when the process $V_t$ hits a time-dependent barrier of exponential shape $B(t)=K\exp(\gamma t)$, with $\gamma \in \mathbb{R}$ and $0<K<V_0$. In this case, explicit formulae for the probability density and all the moments of the FPT $\tau:=\inf\lbrace t\geq 0\vert V_t=B(t)\rbrace$ are well-known; see e.g.~\cite{bielecki2001credit}.

In the financial literature, in particular in that related to structural models of credit risk, $V_t$ is identified with the value of a firm (liquid) assets that can be used to face debt. The process $B(t)$ represents, for each $t$, the infimum of the values that $V_t$ can reach without entering default. It is, therefore, related to debt at a maturity time $T_M$, since $B(T_M)$ represents the value of faced debt. Merton original model considers that default can only occur at time $T_M$, and it happens whenever $V_{T_M} \le B(T_M)$, see~\cite{hull2004merton}. The Black-Cox implementation of the moving barrier $B(t)$ captures the market-observed fact that default can take place any time before maturity~\cite{credit2004bielecki}.

After the changes of variables $X_t := V_t/B(t)$ and $Y_t := \ln(V_t/B(t))$, the FPT of $V_t$ through barrier $B(t)$ will be given by the first time that $X_t=1$ or, equivalently, $Y_t=0$, that is $\tau=\inf\lbrace t\geq 0\vert\sigma W_t=\gamma t-(\mu-\sigma^2/2)t-q\rbrace$, where $q:=\ln(V_0/K)$. So the problem reduces to the well-known FPT of Brownian motion through a linearly moving barrier, see Section 3.5.C in~\cite{karatzas1991brownian}. By It\^o lemma
$$dY_t\;=\;[\mu -\sigma^2/2 - \ln(B(t))'] dt + \sigma dW_t,$$
so the process $Y_t$ is a martingale if $\gamma =\mu -\sigma^2/2$, a supermartingale if $\gamma >\mu -\sigma^2/2$, and a submartingale if $\gamma<\mu -\sigma^2/2$. In particular:
\begin{itemize}
    \item If $\gamma \ge \mu -\sigma^2/2$ the passage event happens in finite time with probability 1, and if the inequality is fulfilled, then in finite mean time, i.e. in a well-defined time scale (in case of equality, it occurs in infinite mean time).
    \item If $\gamma <\mu -\sigma^2/2$, then the passage event occurs with a probability that is positive but strictly less than 1.
\end{itemize}
Respectively in financial terms:
\begin{itemize}
    \item If the debt grows fast enough compared to the value of the company (balanced with its volatility), the company will eventually enter bankruptcy, even if this value grows on average. However, this does not mean that it will enter bankruptcy before the maturity time of the debt. Therefore, the previous discussion tacitly assumes that $B(t)$ is a valid default-inducing barrier for all future times. In other words, there is no realistic maturity time, as debt is always refinanced or reacquired. Or, at least, if such a maturity time exists, it can be considered to happen in the distant future.
    \item If the value of the firm grows faster than the debt (or, more precisely, than the safety covenants prior to debt payment establishing the ability of the firm to pay future debt), the company might never enter into bankruptcy. The same happens if the debt decreases (or safety covenants decrease) fast enough, even if the firm value decreases on average. Note that this actually means that it might either eventually become bankrupt or survive indefinitely. Still, nothing implies that default will not take place before maturity, so we are effectively assuming an infinitely distant maturity date, as in the previous paragraph (or, in practical terms, that maturity will happen long after the other characteristic time scales of interest).
\end{itemize}

Partially inspired by the statistical physics of phase transitions, we might think of this bifurcation as a critical phenomenon~\cite{quantum2021zinn}. In the same line, one might wonder if such a classification can be extended to more generally moving barriers.
Specifically, if the absolute difference $\vert\ln(B(t)/K)-(\mu-\sigma^2/2)t\vert$ grows slowlier than linearly, then such a case does not trivially reduce to the previous ones. In other words, for barrier types of the form $B(t)=K\exp(\gamma t+\tilde{B}(t))$ with $\vert\tilde{B}(t)\vert$ growing sublinearly, the $\gamma t$ term dominates asymptotically whenever $\gamma \neq \mu-\sigma^2/2$, and its qualitative behavior reduces to that of the Black-Cox model. On the other hand, if $\gamma=\mu-\sigma^2/2$, the asymptotic shape of $\tilde{B}(t)$ becomes crucial to determine the almost sure finiteness of the FPT as well as the finiteness of its mean.

Regarding credit risk assessment (for long maturities), those cases for which $\gamma=\mu-\sigma^2/2$ are the most undecidable, as small perturbations on the parameter calibration tilt the prediction on their ability to face debt one way or another. This is why these critical cases would benefit the most from a finer mathematical characterization. In such a case, the $\tilde{B}(t)$ term in the exponential of the barrier can be understood as the introduction in the model of a finer debt structure, safety covenants, or macroeconomic features than those considered in Black-Cox model. This new term improves the decidability on whether the firm will be able to face debt or not, helping to default the firm consequently.

Another justification for this structure is the inclusion of the role of monetary inflation in the Black-Cox model. To include its effect in this mathematical model, we will mirror the developments in the study of diffusion processes in the expanding universe~\cite{aevy2019,eyav2018,veay2018,yae2016}. Although cosmic and monetary inflation are, obviously, different things, we can still try to transfer mathematical knowledge between both. Thus, as we already got some inspiration from the physics of phase transitions, now we will get inspired, in part, by cosmology. Monetary inflation can have a dual effect: if the contraction of demand is smaller than the increase in prices, $V_t$ grows with inflation (usually the case for small steady inflation), while faster inflation can cause a contraction of demand that outweights the increase in revenue due to higher prices, reducing $V_t$. We will say that the first case has a positive \textit{effect of inflation}, while the second case has a negative one. Further models, as the generalized Merton model, include a term in the drift of the stochastic differential equation for $V_t$ that is independent of $V_t$ and represents fixed costs of the firm. These also increase with inflation. For the sake of simplicity, we will not consider the effect of fixed costs. Instead, we just consider that the value process $V_t$ is affected by the inflationary transformation
$$
V_t \longrightarrow A(t) V_t.
$$
We take the monetary scale factor $A(t)$ to be deterministic, continuously differentiable (for the time being, this assumption will be eventually relaxed), positive, and such that $A(0)=1$. Such a function is able to capture both contractive and expansive effects of inflation on the firm asset value $V_t$. If we assume that the debt is unaffected by inflation\footnote{Note that some financial scenarios may require more complex debt models that consider interest rates to be coupled to inflation, in turn affecting the form of the barrier $B(t)$.}, then, by equation (32) in~\cite{yae2016}, the stochastic differential equation for $V_t$ becomes
$$ \nonumber
dV_t = [\mu + \ln(A(t))'] V_t dt + \sigma V_t dW_t
\;\Rightarrow\;
V_t = V_0 \exp[(\mu - \sigma^2/2)t + \sigma W_t] \frac{A(t)}{A(0)}.
$$
The factor $\ln(A(t))'$ is the economic analogue of the Hubble parameter in cosmology; note that we left $A(0)$ (which value is 1) in this formula to highlight the parallelism with the cosmological problem studied in~\cite{yae2016}. If the \textit{effect of inflation} is positive, it is qualitatively equivalent to an increase of the drift of the process $V_t$, so it improves the chances of survival for a barrier with fixed $\gamma$. This mirrors the fact that inflation helps companies overcome their debts, even  achieving survival in cases where, without inflation, a firm would become bankrupt. On the other hand, in a deflationary (more precisely, \textit{negative effect of inflation}) scenario, $A(t)$ acts qualitatively as a decrease in the drift of $V_t$, augmenting the chances of default, even for firms that would avoid bankruptcy without deflation. This analysis deserves further characterization, even more so for the critical case $\gamma=\mu-\sigma^2/2$.

Note that this way of introducing inflation matches perfectly with the previous discussion, as the change
$$
Y_t := \frac{V_t}{B(t)/A(t)}
$$
reduces again the first hitting time problem to a first passage through the origin. Even if we assume the standard Black-Cox structure for the debt (or its safety covenants), i.e. the exponential form of $B(t)$, inflation introduces another time dependence that might be incorporated through the relation $(\mu-\sigma^2/2-\gamma)t-\tilde{B}(t)=\ln(A(t))$\footnote{Note that even the increase in fixed costs due to inflation, given that it is independent of $V_t$, can be included here, given the potential generality of the time dependence of $\tilde{B}(t)$.}; that is, the problem is reduced to study a general barrier form. Note that this mathematical formulation can accommodate different financial phenomena: if safety covenants were introduced through a different form of the barrier, prior to inflation or without considering inflation, the mathematical formalization remains unchanged, only the financial interpretation changes.

%The analogy can also be made with the inflationary model and the FPT through a constant level $K<V_0$, simply by setting $\gamma=0$ and absorbing all the time dependence of the barrier in the inflation function $A(t)$, which allows us to bring the conclusions derived in \cite{yae2016} to the present case just by exploiting this connection.  Lastly, realize that, by leaving $\tilde{B}(t)$ as general as possible (but, for the time being, continuous), we have a lot of room for modeling  the time dependence of inflation. 

%Independently of the chosen financial interpretation for the term $\tilde{B}(t)$ in the barrier, it is clear that its effect will be more relevant in the critical case when $\gamma=(\mu-\sigma^2/2)t$, and the analysis will be more interesting for $\vert\tilde{B}(t)\vert<ct+b\;\forall t\geq 0$ for some positive constants $c$ and $b$, as otherwise the qualitative features of the model will be similar to those of Black-Cox model. In this way, we hope to help improve credit risk assessment for those critical cases where Black-Cox model is most undecidable and small less-than-linear perturbations in the exponent of the barrier or value process (such as inflation) tilts the firm's ability to face debt one way or another.

%Thus, hoping to have justified the financial interest of studying these more general barriers, in the following section we will explore the FPT statistics depending on the particular shape of $\tilde{B}(t)$, identifying some finer critical  cases when $\tilde{B}(t)$ depends on time more slowly than linearly. 

In general, the motivation for including diverse barrier forms comes from different sources. Apart from the effect of inflation, one can also consider the effect of extreme weather events. Both suggest the inclusion of oscillating barriers, which we will pay attention to in the following. Their precise shape would need to be deduced from economic, geophysical, and other factors. In the remainder of this work we will skip these to concentrate on constructing a general mathematical framework.

\section{The finer structure of the critical case}

\label{sec:main_FPT}
As described in the previous section, our aim is to generalize the Black-Cox model by considering a general barrier such that $t \mapsto B(t)$ is deterministic, continuous, and positive; in addition, we will assume that $B(0) < V_0$. Morally, we can think of barriers of the form
$$B(t)=K\exp((\mu-\sigma^2/2)t+\tilde{B}(t)).$$
Although our analysis will be general, we pay special attention to $\vert\tilde{B}(t)\vert$ growing sublinearly, which is the critical case. We will recursively employ the changes $X_t := V_t/B(t)$ and $Y_t := \ln(V_t/B(t))$, so the FPT of $V_t$ through $B(t)$ reduces to the first time that $X_t=1$ or $Y_t=0$. Since
$Y_t=q+\sigma W_t-\tilde{B}(t)$ with $q:=\ln(V_0/K)$, then the FPT is equivalently defined as
$$
\tau:=\inf\lbrace t\geq 0\vert V_t=B(t)\rbrace=\inf \left\lbrace t\geq 0 \Bigg\vert W_t=\frac{\tilde{B}(t)-q}{\sigma} \right\rbrace.
$$
Our goal is to identify the barriers that distinguish between a finite and an infinite first moment of the FPT and those that influence whether the process has a positive or null probability of surviving for all times.

We begin with two lemmata that will support our upcoming results.

\begin{lemma}\label{llinf}
Let $f,g:\mathbb{R} \longrightarrow \mathbb{R}$ be such that $\limsup_{t \to \infty} f < \infty$ and $\liminf_{t \to \infty} g < \infty$, then
$$
\liminf_{t \to \infty} (f+g) \le \limsup_{t \to \infty} f + \liminf_{t \to \infty} g.
$$
\end{lemma}

\begin{proof}
We start with the case in which $\limsup_{t \to \infty} f > -\infty$.
Then, for any $h:\mathbb{R} \longrightarrow \mathbb{R}$, we have
\begin{eqnarray}\nonumber
\liminf_{t \to \infty} (h-f) &\ge& \liminf_{t \to \infty} h + \liminf_{t \to \infty} (-f) \\ \nonumber
&=& \liminf_{t \to \infty} h - \limsup_{t \to \infty} f,
\end{eqnarray}
thus
$$\nonumber
\liminf_{t \to \infty} h \le \liminf_{t \to \infty} (h-f) + \limsup_{t \to \infty} f;
$$
now take $h=f+g$ to conclude.

If $\limsup_{t \to \infty} f = -\infty$ then $\lim_{t \to \infty} f = -\infty$, so we have
$$
\limsup_{t \to \infty} f + \liminf_{t \to \infty} g =-\infty;
$$
but also
$$\nonumber
\liminf_{t \to \infty} (f+g) = \lim_{t \to \infty} f + \liminf_{t \to \infty} g = -\infty,
$$
so the statement follows.
\end{proof}

\begin{lemma}
\label{lemma:rem2barriers}
Let $\tau_i$ be the FPT of an almost surely continuous stochastic process $C_t$ through a barrier given by the continuous function $B_i(t)$, $i=1,2$, i.e. 
$$
\tau_i = \inf \{t \ge 0 : C_t = B_i(t)\}.
$$
If $B_1(t) \geq B_2(t)\;\forall \, t \geq 0$ and $C_0 > B_1(0)$, then $\tau_1 \leq \tau_2$ almost surely and $\mathbb{E}(\tau_1) \leq \mathbb{E}(\tau_2)$ (be they finite or not). 
\end{lemma}

\begin{proof}
The statement is, clearly, a direct consequence of the continuity of $B_i(t)$, $i=1,2$, and the almost sure continuity of $C_t$ through the intermediate value theorem.
\end{proof}

The next result is the main theorem of this section. It establishes conditions for the almost sure finiteness of the FTP and the finiteness of its mean, providing a finer classification than that of the previous section. Although it will be subsequently refined, it might already be sufficiently detailed for some purposes.

\begin{theorem}\label{mainfpt}
Consider the stochastic differential equation
$$
dV_t=\mu V_t dt + \sigma V_t dW_t, \qquad \left. V_t \right|_{t=0}= V_0,
$$
with $\mu \in \mathbb{R}$ and $\sigma, V_0>0$. Let the function
\begin{eqnarray}\nonumber
B:[0,\infty) &\longrightarrow& (0,\infty) \\ \nonumber
t &\longmapsto& B(t),
\end{eqnarray}
be such that $B(t) \in \mathcal{C}([0,\infty))$, $K \equiv B(0) < V_0$, and $B(t)>0$ for all $t \ge 0$. Then, if
$$
\tau := \inf \{t \ge 0 : V_t = B(t)\},
$$
it holds that:
\begin{enumerate}
\item[(a)] Whenever either
$$
\liminf_{t \to \infty} B(t) \Big/ \left\{ K \exp \left[(\mu-\sigma^2/2)t - \sigma \sqrt{2t \ln(\ln(t))}\right] \right\} >1,
$$
or
$$
\limsup_{t \to \infty} B(t) \Big/ \left\{ K \exp \left[(\mu-\sigma^2/2)t + \sigma \sqrt{2t \ln(\ln(t))}\right] \right\} >1,
$$
then $\tau < \infty$ almost surely.
\item[(b)] Whenever
$$
B(t) \ge K \exp \left[(\mu-\sigma^2/2)t + \alpha \sqrt{t} \, \right]\;\forall \, t > 0,
$$
for some $\alpha > \sigma$, then $\mathbb{E}(\tau)<\infty$; while if
$$
B(t) \le K \exp \left[(\mu-\sigma^2/2)t + \sigma \sqrt{t} \, \right]\;\forall \, t > 0, 
$$
then $\mathbb{E}(\tau)=\infty$.
\end{enumerate}
\end{theorem}

\begin{proof}
The proof of this theorem is broken down into two parts. Firstly, we address the almost sure finiteness of $\tau$ via the law of the iterated logarithm~\cite{karatzas1991brownian,khintchine1924satz,brownian2010morters,brownian2014sp}. Secondly, the finiteness of the mean FPT is approached by means of the optional stopping theorem~\cite{handbook2012borodin,karatzas1991brownian,brownian2010morters,brownian2014sp}. Before starting, we note the obvious fact that $q:=\ln(V_0/K)>0$ under the present assumptions.

Let us start proving (a).
It is a classical result that the stochastic differential equation in the statement has a unique solution that is both strong and global, and given by the geometric Brownian motion
$$
V_t=V_0 \exp \left[ (\mu - \sigma^2/2)t + \sigma W_t \right],
$$
see for instance~\cite{oksendal2013}. The almost sure positivity and continuity of this stochastic process are direct consequences of this formula and the properties of Brownian motion. Therefore, the FPT can be written as
\begin{eqnarray}\nonumber
\tau &=& \inf \{t \ge 0 : \ln(V_t) = \ln(B(t))\}\\ \nonumber
&=& \inf \{t \ge 0 : \ln(V_0) + (\mu - \sigma^2/2)t + \sigma W_t - \ln(B(t)) =0\}.
\end{eqnarray}
Now, with probability one it holds that
\begin{eqnarray}\nonumber
\liminf_{t \to \infty} \frac{\ln(V_t)-(\mu - \sigma^2/2)t}{\sqrt{2t \ln(\ln(t))}} &=& \liminf_{t \to \infty} \frac{\ln(V_0) + \sigma W_t}{\sqrt{2t \ln(\ln(t))}}=-\sigma,\\ \nonumber
\limsup_{t \to \infty} \frac{\ln(V_t)-(\mu - \sigma^2/2)t}{\sqrt{2t \ln(\ln(t))}} &=& \limsup_{t \to \infty} \frac{\ln(V_0) + \sigma W_t}{\sqrt{2t \ln(\ln(t))}}=\sigma,\\ \nonumber
\liminf_{t \to \infty} \frac{\ln(B(t))-(\mu - \sigma^2/2)t}{\sqrt{2t \ln(\ln(t))}} &>& -\sigma,
\\ \nonumber
\limsup_{t \to \infty} \frac{(\mu - \sigma^2/2)t-\ln(B(t))}{\sqrt{2t \ln(\ln(t))}} &<& \sigma,
\end{eqnarray}
where we have used the law of the iterated logarithm and the assumption
$$
\liminf_{t \to \infty} B(t) \Big/ \left\{ K \exp \left[(\mu-\sigma^2/2)t - \sigma \sqrt{2t \ln(\ln(t))}\right] \right\} >1
$$
respectively.
Therefore, by Lemma~\ref{llinf} and almost surely, we find that
\begin{eqnarray}\nonumber
&& \liminf_{t \to \infty} \frac{\ln(V_t)-\ln(B(t))}{\sqrt{2t \ln(\ln(t))}} \le \\ \nonumber
&& \liminf_{t \to \infty} \frac{\ln(V_t)-(\mu - \sigma^2/2)t}{\sqrt{2t \ln(\ln(t))}} + \limsup_{t \to \infty} \frac{(\mu - \sigma^2/2)t-\ln(B(t))}{\sqrt{2t \ln(\ln(t))}}<0.
\end{eqnarray}
Consequently, $\liminf_{t \to \infty} \ln(V_t)-\ln(B(t))=-\infty$, so, by continuity of $B(t)$ and the intermediate value theorem, the result follows.

Instead of the previous assumption, assume now that the hypothesis
$$
\limsup_{t \to \infty} B(t) \Big/ \left\{ K \exp \left[(\mu-\sigma^2/2)t + \sigma \sqrt{2t \ln(\ln(t))}\right] \right\} >1
$$
holds. In such a case we have the limiting behavior:
\begin{eqnarray}\nonumber
\limsup_{t \to \infty} \frac{\ln(B(t))-(\mu - \sigma^2/2)t}{\sqrt{2t \ln(\ln(t))}} &>& \sigma,
\\ \nonumber
\liminf_{t \to \infty} \frac{(\mu - \sigma^2/2)t-\ln(B(t))}{\sqrt{2t \ln(\ln(t))}} &<& -\sigma.
\end{eqnarray}
Then, again by Lemma~\ref{llinf} and almost surely, we get that
\begin{eqnarray}\nonumber
&& \liminf_{t \to \infty} \frac{\ln(V_t)-\ln(B(t))}{\sqrt{2t \ln(\ln(t))}} \le \\ \nonumber
&& \limsup_{t \to \infty} \frac{\ln(V_t)-(\mu - \sigma^2/2)t}{\sqrt{2t \ln(\ln(t))}} + \liminf_{t \to \infty} \frac{(\mu - \sigma^2/2)t-\ln(B(t))}{\sqrt{2t \ln(\ln(t))}}<0.
\end{eqnarray}
The consequence of this inequality, as before, is that $\liminf_{t \to \infty} \ln(V_t)-\ln(B(t))=-\infty$; thus, once more by continuity of $B(t)$ and the intermediate value theorem, the statement follows. This completes the proof of part (a) of this theorem.

To prove part (b), by Lemma~\ref{lemma:rem2barriers}, it is enough to compute the mean FPT of $V_t$ through
$$
B(t) = K \exp \left[(\mu-\sigma^2/2)t + \alpha \sqrt{t} \, \right]
$$
for $\alpha \ge \sigma$. In this case, the FPT reads
$$
\tau=\inf \left\{t \ge 0 : \ln(V_0/K) + \sigma W_t = \alpha \sqrt{t} \, \right\},
$$
which is, by part (a) of this proof, finite almost surely even for any $\alpha \in \mathbb{R}$.
This translates to
\begin{eqnarray}\nonumber
\alpha \sqrt{\tau} &=& \ln(V_0/K) + \sigma W_\tau, \\ \nonumber
\alpha \sqrt{t \wedge \tau} &\le& \ln(V_0/K) + \sigma W_{t \wedge \tau},
\end{eqnarray}
so, in consequence
$$
\alpha^2 \, t \wedge \tau \le (\ln(V_0/K))^2 + \sigma^2 W_{t \wedge \tau}^2 + 2 \ln(V_0/K) \, \sigma W_{t \wedge \tau}.
$$
Now we obtain
\begin{eqnarray}\nonumber
\alpha^2 \, \mathbb{E} (t \wedge \tau) &\le& (\ln(V_0/K))^2 + \sigma^2 \, \mathbb{E}(W_{t \wedge \tau}^2) + 2 \ln(V_0/K) \, \sigma \, \mathbb{E} (W_{t \wedge \tau}) \\ \nonumber
&=& (\ln(V_0/K))^2 + \sigma^2 \, \mathbb{E}(t \wedge \tau),
\end{eqnarray}
where we have used, in the first step, the linearity of the expectation and, in the second, the optional stopping theorem along with the fact that $W_t$ and $W_t^2 -t$ are martingales. This, for $\alpha > \sigma$, yields
$$
\mathbb{E} (t \wedge \tau) \le \frac{(\ln(V_0/K))^2}{\alpha^2 - \sigma^2}.
$$
This bound, in turn, implies the summability of $\tau$:
$$
\mathbb{E} (\tau) = \mathbb{E} \left( \lim_{t \to \infty} t \wedge \tau \right) = \lim_{t \to \infty} \,
\mathbb{E} (t \wedge \tau) \le \frac{(\ln(V_0/K))^2}{\alpha^2 - \sigma^2},
$$
where we have used the monotone convergence theorem in the second equality.
Additionally, again by the optional stopping theorem:
$$
\sup_{t \ge 0} \mathbb{E}(W_{t \wedge \tau}^2) = \sup_{t \ge 0} \mathbb{E}(t \wedge \tau) \le \frac{(\ln(V_0/K))^2}{\alpha^2 - \sigma^2} < \infty;
$$
then, by the Doob martingale convergence theorem~\cite{revuz2013continuous} (and the martingality of $W_t$), we conclude that $W_{t \wedge \tau} \to W_\tau$ as $t \to \infty$ both almost surely (by the almost sure finiteness of $\tau$) and in $L^2(\Omega)$, and hence in $L^1(\Omega)$. Therefore
$$
\mathbb{E}(W_{\tau}^2) = \lim_{t \to \infty} \mathbb{E}(W_{t \wedge \tau}^2) = \mathbb{E}(\tau), \qquad \mathbb{E}(W_{\tau}) = \lim_{t \to \infty} \mathbb{E}(W_{t \wedge \tau}) = 0.
$$
We also have
$$
\alpha^2 \, \tau = (\ln(V_0/K))^2 + \sigma^2 W_{\tau}^2 + 2 \ln(V_0/K) \, \sigma W_{\tau},
$$
so by the linearity of the expectation
$$
\alpha^2 \, \mathbb{E}(\tau) = (\ln(V_0/K))^2 + \sigma^2 \, \mathbb{E}(\tau). 
$$
For any $\alpha > \sigma$, we conclude that
$$
\mathbb{E}(\tau) = \frac{(\ln(V_0/K))^2}{\alpha^2 - \sigma^2} < \infty.
$$
Finally, for $\alpha = \sigma$, by Lemma~\ref{lemma:rem2barriers} we deduce
$$
\mathbb{E}(\tau) \ge \sup_{\alpha > \sigma} \frac{(\ln(V_0/K))^2}{\alpha^2 - \sigma^2} = \infty,
$$
and the statement follows.
\end{proof}

\begin{remark}
The assumption $B(t)>0$ for all $t \ge 0$ comes from the fact that the stochastic process $\{V_t; t \ge 0\}$ is positive almost surely; so it is possible to let $B(t)$ to take null or negative values (at least mathematically speaking), but of limited interest. 
\end{remark}

This theorem provides a closed formula for the mean FPT for a barrier component $\tilde{B}(t)\propto \sqrt{t}$. That, in turn, implies an upper bound for the mean FPT for barriers lower bounded by such a function. The precise statement is as follows.

\begin{corollary}
\label{cor:mainfpt}
The mean value of the stopping time
$$
\tau = \inf \{t \ge 0 : V_t = B(t)\}
$$
through the barrier
$$
B(t) = K \exp \left[(\mu-\sigma^2/2)t + \alpha \sqrt{t} \, \right],
$$
for any $\alpha > \sigma$, is given by the explicit formula:
$$
\mathbb{E}(\tau) = \frac{(\ln(V_0/K))^2}{\alpha^2 - \sigma^2}.
$$
Moreover, the mean FPT for a barrier that fulfills
$$
B(t) \ge K \exp \left[(\mu-\sigma^2/2)t + \alpha \sqrt{t} \, \right],
$$
for all $t > 0$ and some $\alpha > \sigma$, obeys the explicit bound
$$
\mathbb{E}(\tau) \le \frac{(\ln(V_0/K))^2}{\alpha^2 - \sigma^2}.
$$
\end{corollary}

\begin{proof}
The statement is a direct consequence of the proof of Theorem~\ref{mainfpt} and the application of Lemma~\ref{lemma:rem2barriers}.
\end{proof}
The change $(V_t,B(t))\rightarrow (W_t,\tilde{B}(t))$, introduced in the previous section, has been used tacitly in the proof of Theorem~\ref{mainfpt}. With this language, the FPT $\tau:=\inf\lbrace t\geq 0\vert \sigma W_t=\tilde{B}(t)-q\rbrace$ is almost surely finite subject to the positivity of the long-time behavior of the infimum/supremum of  $\tilde{B}(t)\pm\sigma\sqrt{2t\ln(\ln(t))}$ respectively. Similarly, its first moment is finite if $\tilde{B}(t) \geq \alpha\sqrt{t}$ for some $\alpha > \sigma$ and all $t\geq0$. This notation will be extensively used in the next section, which is devoted to a more detailed analysis of the properties of the survival probability.

\begin{remark}
Some readers might wonder why we have not used the changes for $X_t$ and $Y_t$ to work directly with either of the stochastic differential equations
\begin{eqnarray} \nonumber
dX_t &=& [\mu - \ln(B(t))'] X_t dt + \sigma X_t dW_t, \\ \nonumber
dY_t &=& [\mu -\sigma^2/2 - \ln(B(t))'] dt + \sigma dW_t.
\end{eqnarray}
If we think of $(\cdot)'$ as being a classical derivative, these equations
only make sense if the barrier function $B(t) \in \mathcal{C}([0,\infty)) \cap \mathcal{C}^1((0,\infty))$. Alternatively, if we consider it to be a Radon-Nikodym derivative, then it would be enough to have $B(t) \in \mathcal{AC}([0,\infty))$ to obtain {\it bona fide} It\^o stochastic differential equations. However, we have assumed the lesser regularity $B(t) \in \mathcal{C}([0,\infty))$. This permits to include functions  as regular as Brownian paths, or even rougher, contrary to the two previous assumptions. In particular, we can select barrier functions that are solutions to rough~\cite{lyons1998differential} or stochastic differential equations~\cite{oksendal2013}, provided they are independent of the filtration generated by $W_t$. This highlights the advantage of the present approach.
\end{remark}

\section{The even finer structure of criticality}\label{sec:critical_prob}

The aim of this section is to further clarify the properties of $\mathbb{P}(\tau<\infty)$. In particular, note that Theorem~\ref{mainfpt} does not cover those cases in which the barriers fulfill

$$
\liminf_{t \to \infty}\tilde{B}(t)- \sigma \sqrt{2t \ln(\ln(t))} \le 0\;\text{ or }\;
\limsup_{t \to \infty} \tilde{B}(t)+\sigma \sqrt{2t \ln(\ln(t))} \le 0.
$$
Along with these specific asymptotic issues, we will also state results that concern general continuous barriers. We start by proving the fact that survival cannot happen almost surely.

\begin{proposition}\label{prop:tlti}
Let $V_t$ be the unique solution to the SDE
$$
dV_t=\mu V_t dt + \sigma V_t dW_t, \qquad \left. V_t \right|_{t=0}= V_0,
$$
with $\mu \in \mathbb{R}$ and $\sigma, V_0>0$, and the barrier function
\begin{eqnarray}\nonumber
B:[0,\infty) &\longrightarrow& (0,\infty) \\ \nonumber
t &\longmapsto& B(t),
\end{eqnarray}
be such that $B(t) \in \mathcal{C}([0,\infty))$ and $K \equiv B(0)<V_0$, but arbitrary otherwise. Then, for the FPT
$$
\tau := \inf \{t \ge 0 : V_t = B(t)\},
$$
it holds that $\mathbb{P}(\tau < \infty) >0$.
\end{proposition}

\begin{proof}
As already noted in the proof of Theorem~\ref{mainfpt}, by the monotony of the logarithm, the FPT can be written as
\begin{eqnarray}\nonumber
\tau &=& \inf \{t \ge 0 : \ln(V_t) = \ln(B(t))\}\\ \nonumber
&=& \inf \{t \ge 0 : \ln(V_0) + (\mu - \sigma^2/2)t + \sigma W_t = \ln(B(t)) \} \\ \nonumber
&=& \inf \{t \ge 0 : W_t = \ln(B(t))/\sigma - \left[(\mu - \sigma^2/2)t + \ln(V_0) \right]/\sigma \}.
\end{eqnarray}
By the continuity of $B(t)$, defining $q:=\ln(V_0/K)$, the function
$$
\hat{B}(t):=(\tilde{B}(t) -q)/\sigma= \ln(B(t))/\sigma - \left[(\mu - \sigma^2/2)t + \ln(V_0) \right]/\sigma
$$
is continuous too and such that $\hat{B}(0) = \ln(K/V_0)/\sigma <0$. The support theorem of Brownian motion~\cite{brownian2010morters} guarantees that
\begin{equation}\label{eq:stbm}
\mathbb{P} \left( \left\{ \sup_{0 \le t \le T} |W_t - f(t)| < \epsilon \right\} \right) >0\;\,\forall\,\epsilon>0,
\end{equation}
provided $f:[0,T] \longrightarrow \mathbb{R}$, $T>0$, is continuous and such that $f(0)=0$ (but otherwise arbitrary).
Since $\hat{B}(t)$ is continuous, it attains its minimum (which is obviously negative) in any interval $[0,T]$. Fix $\epsilon, T >0$ and build the function $f(t)$ in~\eqref{eq:stbm} as
$$
f(t) := \left( \min_{0 \le s \le T} \hat{B}(s) -\delta \right) \frac{t}{T}
$$
for any $\delta > 2\epsilon$; so that
\begin{eqnarray}\nonumber
f(T) + \epsilon < \min_{0 \le s \le T} \hat{B}(s) -\epsilon.
\end{eqnarray}
Then we have the inclusion of events
\begin{equation}\nonumber
\left\{ \sup_{0 \le t \le T} |W_t - f(t)| < \epsilon \right\} \subseteq\left\{ |W_T - f(T)| < \epsilon \right\} 
\subseteq \left\{ W_T - f(T) < \epsilon \right\}
\subseteq\omega_\epsilon(T)
\end{equation}
for $\omega_\epsilon(T):=\left\{ W_T < \min\limits_{0 \le s \le T} \hat{B}(s) -\epsilon \right\}$, and, thus, we obtain:
$$\nonumber
\mathbb{P} \left(\omega_\epsilon(T)\right) >0,\;\text{with $\epsilon >0$.}
$$
 Defining $\omega_{c}(t):=\left\{ \text{$W_t$ is continuous} \right\}$, by the law of total probability
\begin{equation}\nonumber
0 < \mathbb{P} \left( \omega_\epsilon(T) \right)
= \mathbb{P} \left( \omega_\epsilon(T) \cap \omega_{c}(t) \right)+
\mathbb{P} \left( \omega_\epsilon(T) \cap \omega_{c}^c(t) \right)
=\mathbb{P} \left( \omega_\epsilon(T)\cap \omega_{c}(t)\right),
\end{equation}
since $
0 \le \mathbb{P} \left(\omega_\epsilon(T)\cap \omega_{c}^c(t) \right)
\le \mathbb{P} \left( \omega_{c}^c(t) \right)=0$. And so, by the intermediate value theorem, $W_t$ and $\hat{B}(t)$ (equivalently, $V_t$ and $B(t)$) cross at some $t \in (0,T)$ with positive probability.
\end{proof}

\begin{remark}\label{rem:nfpt}
Note that if we modified the definition of the FPT by
$$
\tau := \inf \{ 0 \le t \le T : V_t = B(t)\}
$$
with the convention $\inf \emptyset = \infty$ and
for any $T >0$, then the proof would still directly imply that $\mathbb{P}(\tau < \infty) >0$.
\end{remark}

We continue stating an intuitive fact.

\begin{corollary}\label{obviousfact}
Let $V_t$ and $B(t)$ be as in the statement of Proposition~\ref{prop:tlti}. Then $\tau >0$ almost surely and hence $\mathbb{E}(\tau)>0$.
\end{corollary}

\begin{proof}
Suppose on the contrary that $\mathbb{P}(\tau=0)>0$; then
$$
\mathbb{P} \left( \limsup_{T \searrow 0} \,\, \inf \{0 \le t \le T : V_t = B(t)\} \right) >0.
$$
But this, by the proof of Proposition~\ref{prop:tlti}, the continuity of $B(t)$, and the fact that
$$
\mathbb{P} \left( \left\{ \text{$W_t$ is continuous} \right\} \cap \left\{ \text{$W_0$ =0} \right\} \right)=1,
$$
implies that
$$
\mathbb{P} \left( \ln(K/V_0) /\sigma =0 \right) >0,
$$
in contradiction with the assumption $K < V_0$. Therefore
$\mathbb{P}(\tau=0)=0 \Rightarrow \mathbb{P}(\tau > 0)=1$, and then there exists an $\epsilon>0$ such that $\mathbb{P}(\tau > \epsilon)>0$; otherwise we get the contradiction
$$
0=\mathbb{P}(\tau > \epsilon) \underset{\epsilon \searrow 0}{\longrightarrow} \mathbb{P}(\tau > 0)
$$
by the right-continuity of $\mathbb{P}(\tau > \cdot)$. We conclude by the Markov inequality
$$
0 < \mathbb{P}(\tau > \epsilon) \le \frac{\mathbb{E}(\tau)}{\epsilon}.
$$
\end{proof}

\begin{remark}\label{rem:dich}
Note that this corollary along with the statements of Theorem~\ref{mainfpt} and Proposition~\ref{prop:tlti} reduce the critical behavior to the dichotomy
$$
\mathbb{P}(\tau < \infty) < 1 \qquad \text{or} \qquad \mathbb{P}(\tau < \infty) = 1;
$$
at least, if we restrict ourselves to qualitatively determine the finiteness of $\tau$.
\end{remark}

Now we move to the core of the topic of this section and start characterizing those barriers for which $\mathbb{P}(\tau < \infty) < 1$.

\begin{theorem}\label{theor:tlti}
Let $V_t$ and $B(t)$ be as in the statement of Proposition~\ref{prop:tlti}. Moreover, assume that the condition
$$
\limsup_{t \to \infty} B(t) \Big/ \left\{ K \exp \left[(\mu-\sigma^2/2)t - \sigma \sqrt{2t \ln(\ln(t))}\right] \right\} < 1
$$
holds. Then $\mathbb{P}(\tau < \infty) < 1$.
\end{theorem}

\begin{proof}
The condition in the statement implies
$$
\limsup_{t \to \infty} \frac{\ln(B(t))-(\mu - \sigma^2/2)t}{\sigma \sqrt{2t \ln(\ln(t))}} <-1.
$$
Then, arguing as in the proof of Theorem~\ref{mainfpt}, we find
\begin{eqnarray}\nonumber
&& \liminf_{t \to \infty} \frac{\ln(V_t)-\ln(B(t))}{\sigma \sqrt{2t \ln(\ln(t))}} \ge \\ \nonumber
&& \liminf_{t \to \infty} \frac{\ln(V_t)-(\mu - \sigma^2/2)t}{\sigma \sqrt{2t \ln(\ln(t))}} + \liminf_{t \to \infty} \frac{(\mu - \sigma^2/2)t-\ln(B(t))}{\sigma \sqrt{2t \ln(\ln(t))}} = \\ \nonumber
&& \liminf_{t \to \infty} \frac{\ln(V_t)-(\mu - \sigma^2/2)t}{\sigma \sqrt{2t \ln(\ln(t))}} - \limsup_{t \to \infty} \frac{\ln(B(t))-(\mu - \sigma^2/2)t}{\sigma \sqrt{2t \ln(\ln(t))}} >0
\end{eqnarray}
almost surely. Defining
$$Q_t:= \frac{\ln(V_t/B(t))}{\sigma \sqrt{2t \ln(\ln(t))}},$$
the last result means either
\begin{equation}\label{epsilon1}
\liminf_{t \to \infty} Q_t = \epsilon_1,\text{ a.s. for some $\epsilon_1 >0$,}\quad\text{ or else } \quad
\liminf_{t \to \infty} Q_t = \infty,
\end{equation}
almost surely. We start by assuming the existence of such an $\epsilon_1$. Since almost sure convergence implies convergence in probability, from the first equality in equation~\eqref{epsilon1} it follows that
$$
\lim_{t \to \infty} \mathbb{P}\left( \left\lbrace\left| \inf_{s \ge t} Q_s -\epsilon_1 \right| > \epsilon_2 \right\rbrace\right)=0\;\, \forall \, \epsilon_2>0.
$$
Thus, for each $\epsilon_2>0$ there exists a $t_2$ sufficiently large so that
$$
\mathbb{P}\left(\left\lbrace \left| \inf_{s \ge t_2} Q_s-\epsilon_1 \right| \le \epsilon_2 \right\rbrace \right)>0,
\;\text{ implying }\;
\mathbb{P}\left( \left\{ Q_t\ge \epsilon_1 - \epsilon_2 \,\, \forall \, t \ge t_2 \right\} \right)>0,
$$
due to the inclusion of events
$$ \left\lbrace \left| \inf_{s \ge t_2} Q_s -\epsilon_1 \right| \le \epsilon_2 \right\rbrace \subseteq \left\{ Q_t \ge \epsilon_1 - \epsilon_2 \,\, \forall \, t \ge t_2 \right\}.
$$
Finally, taking $t_2$ large enough (so that $\epsilon_2 < \epsilon_1$), we conclude that
$$
\mathbb{P}\left( \left\{ V_t > B(t) \,\, \forall \, t \ge t_2 \right\} \right)>0.
$$
Now assume $\epsilon_1$ does not exist; from the second equality in equation~\eqref{epsilon1} it follows that
$$
\lim_{t \to \infty} \mathbb{P}\left( \left\lbrace \inf_{s \ge t} Q_s > \epsilon_2' \right\rbrace\right)=1\;\, \forall \, \epsilon_2'>0.
$$
As previously, for each $\epsilon_2'>0$ there exists a $t_2'$ sufficiently large so that
$$
\mathbb{P}\left(\left\lbrace \inf_{s \ge t_2'} Q_s > \epsilon_2' \right\rbrace \right)>0,
\;\text{ implying }\;
\mathbb{P}\left( \left\{  Q_t > 0 \,\, \forall \, t \ge t_2' \right\} \right)>0;
$$
so the same conclusion as before holds, and we may drop the primes without loss of generality.

Now, define $\omega_>(t_a,t_b):=\{ V_t > {B}(t) \, \forall \, t\in (t_a,t_b)\}$ and $\omega_>(t_c):=\{ V_{t_c} > {B}(t_c)\}$ for any $t_b > t_a \ge 0$ and $t_c \ge 0$. Then, we have:
$$ 
 \mathbb{P}\left( \left\{ V_t > B(t) \,\, \forall \, t \ge 0 \right\} \right) =
 \mathbb{P}\left( \omega_>(0,t_2) \cap \omega_>(t_2) \cap\omega_>(t_2,\infty) \right)\;\forall \, t_2 >0.
$$
Note that $\omega_>(t_a,t_b)=\{ W_t > \hat{B}(t) \,\, \forall \, t\in (t_a,t_b)\}$ and $\omega_>(t_c)=\{ W_{t_c} > \hat{B}(t_c)\}$ because $V_\cdot > B(\cdot) \Leftrightarrow W_\cdot > \hat{B}(\cdot)$; then:
\begin{eqnarray}\nonumber
\mathbb{P}\left( \omega_>(0,\infty)\right) &=&
\mathbb{P}\left( \omega_>(0,t_2) \cap \omega_>(t_2) \cap \omega_>(t_2,\infty) \right) \\ \nonumber
&=& \mathbb{P}\left(\omega_>(t_2,\infty)| \, \omega_>(t_2) \cap \omega_>(0,t_2) \right)
 \times \mathbb{P}\left( \omega_>(t_2) \cap \omega_>(0,t_2)\right)
\\ \nonumber
&=& \mathbb{P}\left( \omega_>(t_2,\infty) \left| \, \omega_>(t_2) \right. \right) 
\times \mathbb{P}\left( \omega_>(t_2) \cap \omega_>(0,t_2) \right)
\\ \nonumber
&=& \left[ \mathbb{P}\left(\{ W_t > \hat{B}(t) \,\, \forall \, t \ge t_2 \}\right) \Big/ \, \mathbb{P}\left(\omega_>(t_2)\right) \right]
\\ \nonumber
&& \times \mathbb{P}\left( \omega_>(t_2) \cap \omega_>(0,t_2) \right),
\end{eqnarray}
where we have used the definition of conditional probability and the Markovianity of Brownian motion. Since $\hat{B}(\cdot)$ is continuous, it attains its minimum on any closed interval, which in this case we denote
$\mathfrak{m}_2:= \min_{0 \le t \le t_2} \hat{B}(t)$. We thus deduce
\begin{eqnarray}\nonumber
 \mathbb{P}\left(\omega_>(0,\infty) \right) &\ge& \left[ \mathbb{P}\left(\{ W_t > \hat{B}(t) \,\, \forall \, t \ge t_2 \}\right) \Big/ \, \mathbb{P}\left(\left\{ W_{t_2} > \mathfrak{m}_2 \right\}\right) \right]
\\ \nonumber
&& \times \mathbb{P}\left( \omega_>(t_2) \cap \omega_>(0,t_2) \right),
\end{eqnarray}
where $\mathbb{P}\left(\left\{ W_{t_2} > \mathfrak{m}_2 \right\}\right) \in (0,1)$, since $W_{\cdot}$ is normally distributed. 

Now, again by the support theorem of Brownian motion, we know that
$$
\mathbb{P} \left( \left\{ \sup_{0 \le t \le t_2} |W_t - g(t)| < \epsilon_3 \right\} \right) >0
$$
for any $\epsilon_3 >0$, where we can take $g:[0,t_2] \longrightarrow \mathbb{R}$ to be $g(t):= \hat{B}(t)+\ln(V_0/K)/\sigma$, as it is continuous and $g(0)=0$. If we take $\epsilon_3 < \ln(V_0/K)/(2\sigma)$, we might conclude that
$$
\mathbb{P}\left( \omega_>(t_2) \cap \omega_>(0,t_2) \right)>0
\Longrightarrow
\mathbb{P} \left( \left\{ V_t > B(t) \,\, \forall \, t \ge 0 \right\} \right) >0.
$$

Since condition $\mathbb{P} \left( \left\{ V_t > B(t) \,\, \forall \, t \ge 0 \right\} \right) >0$ is complementary to that in the statement, this concludes the proof.
\end{proof}

This theorem has an immediate consequence:

\begin{corollary}\label{cor:nfpt}
Let $V_t$ and $B(t)$ be as in the statement of Proposition~\ref{prop:tlti}. If we define the FPT as
$$
\tau := \inf \{0 \le t \le T : V_t = B(t)\}
$$
with the convention $\inf \emptyset = \infty$, then $\mathbb{P}(\tau < \infty) < 1$ for any $T>0$.
\end{corollary}

\begin{proof}
Immediate from the last part of the proof of Theorem~\ref{theor:tlti}.
\end{proof}

\begin{remark}\label{rem:cornfpt}
We already noted in Remark~\ref{rem:dich} that the first passage problem presents a dichotomy when posed on the real half-line $[0,\infty)$. Nevertheless, we can conclude, from Remark~\ref{rem:nfpt} and Corollary~\ref{cor:nfpt}, that this dichotomy disappears when this problem is posed on a finite time interval $[0,T]$, for any $T>0$, as in such a case we always have 
$
0 < \mathbb{P}(\tau < \infty) < 1.
$
\end{remark}

Our next step is to prove the final result of this section.

\begin{theorem}\label{theor:maincb}
Let $V_t$ and $B(t)$ be as in the statement of Proposition~\ref{prop:tlti}. Moreover, assume that the condition
$$
\liminf_{t \to \infty} B(t) \Big/ \left\{ K \exp \left[(\mu-\sigma^2/2)t - \sigma \sqrt{2t \ln(\ln(t))}\right] \right\} \ge 1
$$
holds. Then $\mathbb{P}(\tau < \infty) = 1$.
\end{theorem}

\begin{proof}
We might assume the equality, as the case of the inequality was already proven in Theorem~\ref{mainfpt}. 
So far we have used extensively the Khintchine law of the iterated logarithm, but this case requires to take into account its convergence rate too. For this, we follow~\cite{KhoshnevisanLevinShi2005}, which extends the classical works~\cite{DarlingErdos1956,Erdos1942LIL}. In particular, Theorem~1.2 in~\cite{KhoshnevisanLevinShi2005} states that the convergence
$$
\lim_{t \to \infty} \frac{\ln(\ln(t))}{\ln(\ln(\ln(t)))} \left[ \sup_{t \le s} \frac{W_s}{\sqrt{2 s \ln(\ln(s))}} -1 \right]=\frac34
$$
takes place almost surely. By symmetry, we also have the almost sure convergence
$$
\lim_{t \to \infty} \frac{\ln(\ln(t))}{\ln(\ln(\ln(t)))} \left[ \inf_{t \le s} \frac{W_s}{\sqrt{2 s \ln(\ln(s))}} +1 \right]=-\frac34.
$$
These imply that, with probability one, it holds that
\begin{eqnarray}\nonumber
&& \liminf_{t \to \infty}
\frac{\ln(\ln(t))}{\ln(\ln(\ln(t)))} \left[
\frac{\ln(V_t)-\ln(B(t))}{\sigma \sqrt{2 t \ln(\ln(t))}} \right] \\ \nonumber
&\le& \liminf_{t \to \infty}
\frac{\ln(\ln(t))}{\ln(\ln(\ln(t)))} \left[
\frac{\ln(V_t)-\ln (\inf_{t \le s} B(s) )}{\sigma \sqrt{2 t \ln(\ln(t))}} \right] \\ \nonumber
&=& \liminf_{t \to \infty} \frac{\ln(\ln(t))}{\ln(\ln(\ln(t)))} \left[ \frac{\ln(V_0/K)/\sigma + W_t}{\sqrt{2t \ln(\ln(t))}} + 1 \right] = -\frac34,
\end{eqnarray}
and, thus, $\liminf_{t \to \infty} \ln(V_t)-\ln(B(t))=-\infty$, so the statement follows from the continuity of $B(t)$, the almost sure continuity of $V_t$, the monotony and continuity of the logarithm, and the intermediate value theorem.
\end{proof}

\begin{remark}
\label{rem:undecidable_cases}
The classification of barriers that we have proven so far is not exhaustive in the sense that it does not cover certain cases. Precisely, for barriers that fulfill simultaneously the conditions
\begin{eqnarray}\nonumber
\limsup_{t \to \infty} B(t) \Big/ \left\{ K \exp \left[(\mu-\sigma^2/2)t - \sigma \sqrt{2t \ln(\ln(t))}\right] \right\} &\ge& 1, \\ \nonumber
\liminf_{t \to \infty} B(t) \Big/ \left\{ K \exp \left[(\mu-\sigma^2/2)t - \sigma \sqrt{2t \ln(\ln(t))}\right] \right\} &<& 1, \\ \nonumber
\limsup_{t \to \infty} B(t) \Big/ \left\{ K \exp \left[(\mu-\sigma^2/2)t + \sigma \sqrt{2t \ln(\ln(t))}\right] \right\} &\le& 1,
\end{eqnarray}
we cannot deduce if $0 < \mathbb{P}(\tau < \infty) < 1$ or $\mathbb{P}(\tau < \infty)=1$ (we only know that $0 < \mathbb{P}(\tau < \infty) \le 1$). Thus, more conditions must be added to ascertain the correct classification in these cases (see Section~\ref{sec:examples}). Besides, this uncertainty cannot take place for barriers with a concretely defined asymptotic behavior, in the sense 
\begin{eqnarray}\nonumber
\limsup_{t \to \infty} B(t) \Big/ \left\{ K \exp \left[(\mu-\sigma^2/2)t - \sigma \sqrt{2t \ln(\ln(t))}\right] \right\} &=& \\ \nonumber
\liminf_{t \to \infty} B(t) \Big/ \left\{ K \exp \left[(\mu-\sigma^2/2)t - \sigma \sqrt{2t \ln(\ln(t))}\right] \right\}. &&
\end{eqnarray}
In this case, we have an exhaustive classification: if
$$
\lim_{t \to \infty} B(t) \Big/ \left\{ K \exp \left[(\mu-\sigma^2/2)t - \sigma \sqrt{2t \ln(\ln(t))}\right] \right\} < 1,
$$
then $0<\mathbb{P}(\tau < \infty) < 1$, while if
$$
\lim_{t \to \infty} B(t) \Big/ \left\{ K \exp \left[(\mu-\sigma^2/2)t - \sigma \sqrt{2t \ln(\ln(t))}\right] \right\} \ge 1,
$$
then $\mathbb{P}(\tau < \infty) = 1$.
\end{remark}

In the next section, we will focus on the first moment of the FPT. We will show that the critical behavior proven here for the survival probability resembles that for the finiteness of the mean FPT, but with different thresholds.

\section{Mean FPTs for general barriers}
\label{sec:mfpts}

The aim of this section is to generalize the previous results concerning mean FPTs to include a wider range of barrier functions. We begin by introducing a refined version of part (b) of Theorem~\ref{mainfpt} and Corollary~\ref{cor:mainfpt}. 

\begin{theorem} 
\label{theor:second_fpt}
Denote $B_c(t,\alpha):=K \exp \left[(\mu-\sigma^2/2)t + \alpha \sqrt{t} \, \right]$ and let $B(t)$ be a barrier function that fulfills the same properties as those stated in Theorem~\ref{mainfpt}. Define $\tau := \inf \{t \ge 0 : V_t = B(t)\}$; then:
\begin{enumerate}[label=(\alph*)]
\item If, for some $\alpha>\sigma$,
$$
\liminf\limits_{t\rightarrow\infty}\dfrac{B(t)}{B_c(t,\alpha)} > 1,
$$
 then $\mathbb{E}(\tau)<\infty$. In particular, with the Gaussian probability distribution function denoted by
$$
\phi_{a,b}(x):=\frac{1}{\sqrt{2\pi} \, b}\exp \left\lbrace -\frac{(x-a)^2}{2 \, b^2} \right\rbrace,
$$
the following upper bound holds
\begin{equation}\nonumber
\mathbb{E}(\tau) \le T + \int_{\alpha\sqrt{T}}^\infty \frac{x^2}{\alpha^2 - \sigma^2} \, \phi_{q,\sigma\sqrt{T}}(x) \, dx,
\end{equation}
where $T:=\inf \left\lbrace t\geq 0 \, \big\vert \, B(s) \ge B_c(s,\alpha) \; \forall \, s \ge t \right\rbrace$ and $q:=\ln(V_0/K)$.
\item If, conversely, 
$$
\limsup\limits_{t\rightarrow\infty}\dfrac{B(t)}{B_c(t,\sigma)} < 1,
$$
then $E(\tau)=\infty$.
\end{enumerate}
\end{theorem}

\begin{proof}
For part \textit{(a)} of the statement, note $\tau<\infty$ almost surely by part \textit{(a)} of Theorem~\ref{mainfpt}. Also, by the continuity of $B(\cdot),B_c(\cdot,\alpha)$, their relative limit behavior, and the intermediate value theorem, $\exists \, T \ge 0$ such that:
$$
T:=\inf \left\lbrace t\geq 0 \, \big\vert \, B(s) \geq B_c(s,\alpha) \; \forall \, s \ge t \right\rbrace<\infty.
$$
This deterministic time denotes the last crossing point of the graphs of $B(\cdot)$ and $B_c(\cdot,\alpha)$ provided it exists, being zero otherwise; by crossing we mean that the relative order of these graphs changes at this point. If $T=0$, this brings us back to the cases analyzed in Theorem~\ref{mainfpt} and Corollary~\ref{cor:mainfpt}, so from now on we assume $T>0$. By the proof of part~\textit{(a)} of Theorem~\ref{mainfpt}, we might rewrite the first passage time as:
\begin{eqnarray}\nonumber
\tau &=& \inf \{t \ge 0 : \ln(V_t) = \ln(B(t))\}\\ \nonumber
&=& \inf \{t \ge 0 : W_t = \ln(B(t))/\sigma - \left[(\mu - \sigma^2/2)t + \ln(V_0) \right]/\sigma \}.
\end{eqnarray}
Now, by the law of the total expectation,
\begin{eqnarray}
\label{eq:total_expectation}
\nonumber
\mathbb{E}(\tau)&=&\mathbb{E}(\tau 1_{\tau\leq T})+\mathbb{E}(\tau 1_{\tau> T})\\ \nonumber
&\leq& T \, \mathbb{P}(\tau\leq T)+\mathbb{E}(\tau 1_{\tau> T}) \\
&=& T+\mathbb{E}[(\tau-T) 1_{\tau> T}],
\end{eqnarray}
and for the last term, again by this law, we can compute:
$$
\mathbb{E}[(\tau-T) 1_{\tau> T}]=\int\limits_{B(T)}^\infty\mathbb{E}(\tau-T\vert\tau>T,V_T=x) \, \mathbb{P}(V_T\in dx, \tau>T),
$$
where the probability density can be decomposed as the product
$$\mathbb{P}(V_T\in dx, \tau>T)=\mathbb{P}(V_T\in dx) \, \mathbb{P}( \tau>T\vert V_T =x).
$$
For brevity, we will denote $L_T(x) dx:=\mathbb{P}(V_T\in dx)$, referencing the lognormal distribution of $V_T$.

We will now focus on $\mathbb{E}(\tau-T\vert\tau>T,V_T=x)$. 
%Let us define $\tilde{B}(t)=\ln\left(\dfrac{B(t)}{K \exp[(\mu-\sigma/2)t]}\right)$
Firstly, consider the random time defined as $\tau^\prime:=(\tau-T)1_{\tau>T}$. That is, $\tau^\prime=0$ if $\tau \leq T$ and else it is the stopping time
\begin{eqnarray}\nonumber
\tau' &=& \inf \{ t \ge 0 : V_{t+T}=B(t+T)\}\\ \nonumber
&=&  \inf \{ t \ge 0 : W_{t+T}-W_{T} = \ln[B(t+T)/V_T]/\sigma - \left[(\mu - \sigma^2/2)t \right]/\sigma \}\\ \nonumber
&=&  \inf \{ t \ge 0 : \sigma(W_{t+T}-W_{T})+q+\sigma W_T = \tilde{B}(t+T)\},
\end{eqnarray}
where we have used that $\tilde{B}(t):=\ln\left(B(t)/[K \exp\{(\mu-\sigma^2/2)t\}]\right)$ and $W_T=[\ln(V_T/V_0)-(\mu-\sigma^2/2)T]/\sigma$. Thus, for each $x>B(T)$ (note that $V_T>B(T)\Leftrightarrow \sigma W_T +q> \alpha\sqrt{T}$, by the definition of $T$) we have:
$$
\mathbb{E}(\tau-T\vert V_T=x,\tau>T) =\mathbb{E}(\tau^\prime\vert V_T=x,\tau>T)=\mathbb{E}(\tau^\prime\vert V_T=x),
$$ 
where the last equality is obtained by the Markovian nature of $W_t$. Besides, $B(t+T)\geq K \exp\{(\mu-\sigma^2/2)(t+T)+\alpha\sqrt{t+T}\}$ by the definition of $T$, implying
\begin{eqnarray}\nonumber
B(t+T) &\geq& K \exp\{(\mu-\sigma^2/2)T\} \\ \nonumber
&& \times
\exp\{(\mu-\sigma^2/2)t+\alpha \sqrt{t+T}\} \\ \nonumber
 &>& K \exp\{(\mu-\sigma^2/2)T\} \\ \nonumber
&& \times
\exp\{(\mu-\sigma^2/2)t+\alpha \sqrt{t}\} \\ \nonumber
&=& B_c(t,\alpha) \exp\{(\mu-\sigma^2/2)T\}.
\end{eqnarray}
By the translation invariance of Brownian motion, the stochastic process
$\{W_{t+T}-W_{T},t \ge 0\}$ is a standard Brownian motion and, by the independence of Brownian increments, it is independent of $W_T$ (and of $V_T$ since $V_T$ is $\sigma(W_T)-$measurable). Therefore, $\tau'$ for a fixed value of $V_T$ (or equivalently $W_T$), falls under the assumptions of Theorem~\ref{mainfpt} and Corollary~\ref{cor:mainfpt}, 
so it is summable and moreover
$$
\mathbb{E}(\tau' \, \vert \, W_T=z) \le \frac{(q +\sigma z)^2}{\alpha^2 - \sigma^2}.
$$
Now, substituting this estimate in equation~\eqref{eq:total_expectation}, we have:
\begin{eqnarray}
\label{eq:for_lemma}
\nonumber
\mathbb{E}(\tau) &\le& T + \int\limits_{B(T)}^\infty\frac{|q+\sigma z(x)|^2}{\alpha^2 - \sigma^2} \, \mathbb{P}(\tau>T\vert V_T=x) \, L_T(x) \, dx \\
&=&
T + \int\limits_{\alpha\sqrt{T}}^\infty\frac{y^2}{\alpha^2 - \sigma^2} \, \mathbb{P}(\tau>T\vert \sigma W_T+q=y) \, \phi_{q,\sigma \sqrt{T}}(y) \, dy,
\end{eqnarray}
where we have used the relations $y=\sigma z(x)+q=\ln(x/K)-(\mu-\sigma^2/2)T$. Note that the change in the integration limits comes from the fact that $V_T>B(T)\Leftrightarrow \sigma W_T+q>\alpha\sqrt{T}$ due to the definition of $T$, which implies $B(T)=B_c(T,\alpha)$. Also, note that, since $V_T$ is lognormally distributed, then $q+\sigma W_T$ is normally distributed with mean $q$ and variance $\sigma^2T$, hence the appearance of $\phi_{q,\sigma\sqrt{T}}(y)$ in the second integrand. To obtain the expression in the statement, simply consider
$$
\mathbb{P}(\tau>T\vert \sigma W_T+q=y)\leq 1
$$
uniformly in $y$.

The proof of part \textit{(b)} is as follows. Firstly, we consider the deterministic crossing time $T_c$ of the barrier $B(t)$ through the critical barrier $B_c(t,\sigma)$. Indeed, by the continuity of $B(\cdot),B_c(\cdot,\sigma)$, their relative limit behavior, and the intermediate value theorem, $\exists \, T_c \ge 0$ such that:
$$
T_c :=\inf \left\lbrace t\geq 0 \, \big\vert \, B(s) \le B_c(s,\sigma) \; \forall \, s \ge t \right\rbrace<\infty.
$$
As before, the case $T_c=0$ is analyzed in Theorem~\ref{mainfpt}, so $T_c>0$ will be assumed. We introduce two random times, $\tau^\prime_c:=(\tau-T_c)1_{\tau>T_c}$ and
$$
\tau_{1}:=\inf \left\lbrace t\geq 0 \, \big\vert \sigma(W_{t+T_c}-W_{T_c})+q+\sigma W_{T_c}=\sigma\sqrt{t+T_c}\right\rbrace.
$$
By the definition of $T_c$, $\tilde{B}(t+T_c)\leq \sigma\sqrt{t+T_c}$. Thus, $\tau^\prime_c \geq \tau_1 \, 1_{\tau>T_c}$ almost surely by Lemma~\ref{lemma:rem2barriers}. We add now another FPT,
$$
\tau_{2}:=\inf \left\lbrace t\geq 0 \, \big\vert \sigma(W_{t+T_c}-W_{T_c})+q+\sigma W_{T_c}=\sigma\left(\sqrt{t}+\sqrt{T_c}\right)\right\rbrace,
$$
and we realize that, by the subadditivity of the square root, i.e. $\sqrt{t+T_c}\leq \sqrt{t}+\sqrt{T_c}\;\forall \, t,T_c\geq 0$, it holds that $\tau_1 \, 1_{\tau>T_c} \geq \tau_2 \, 1_{\tau>T_c}$ almost surely again by Lemma~\ref{lemma:rem2barriers}. As in the first part of the proof, for each fixed $W_{T_c}$ such that $V_{T_c} > B(T_c)$, equivalently $q+\sigma W_{T_c}-\sigma\sqrt{T_c}>0$, $\tau_2$ falls under the assumptions of Theorem~\ref{mainfpt}, but in this case it fulfills the conditions for an infinite mean. Thus, since $\tau^\prime_c \geq \tau_1 \, 1_{\tau>T_c} \geq \tau_2 \, 1_{\tau>T_c}$ almost surely, we find that
$$
\mathbb{E}(\tau^\prime_c\vert V_{T_c}=y,\tau>T_c) \geq \mathbb{E}(\tau_2\vert V_{T_c}=y,\tau>T_c) = \mathbb{E}(\tau_2\vert V_{T_c}=y)= \infty
$$
for every $y>B(T_c)$, where we have employed the Markovianity of Brownian motion.
Using that $1_{\tau>T_c}\leq 1$ and $T_c > 0$, we find
\begin{equation}\nonumber
\mathbb{E}(\tau) \ge \mathbb{E}(\tau 1_{\tau>T_c}) > \mathbb{E}((\tau-T_c) 1_{\tau>T_c}).
\end{equation}
Now, as before, by the law of the total expectation:
\begin{equation}\nonumber
\mathbb{E}(\tau^\prime_c) = \int\limits_{B(T_c)}^\infty \mathbb{E}(\tau^\prime_c \vert V_{T_c}=y,\tau>T_c) \, \mathbb{P}(\tau>T_c\vert V_{T_c}=y) \, \mathbb{P}(V_{T_c} \in dy).
\end{equation}
Since $\mathbb{P}(\tau > T_c) >0$  by Corollary~\ref{cor:nfpt} and
$$
\mathbb{P}(\tau > T_c)=\int\limits_{B(T_c)}^\infty \mathbb{P}(\tau>T_c\vert V_{T_c}=y) \, \mathbb{P}(V_{T_c}\in dy),
$$
we know that $\mathbb{P}(\tau>T_c\vert V_{T_c}=y)$ shares non-empty support with $\mathbb{P}(V_{T_c}\in dy)$ for $y>B(T_c)$. This implies $\mathbb{E}(\tau^\prime_c)=\infty$ and consequently $\mathbb{E}(\tau)=\infty$.
\end{proof}

\begin{corollary}\label{cor:second_fpt}
Let $B_c(t,\alpha)$ and $B(t)$ be as in the statement of Theorem~\ref{theor:second_fpt}. If
$$
\liminf\limits_{t\rightarrow\infty}\dfrac{B(t)}{B_c(t,\alpha)} \ge 1
$$
for some $\alpha>\sigma$, then $\mathbb{E}(\tau)<\infty$ and the same upper bound as that in the statement of Theorem~\ref{theor:second_fpt} holds true.
\end{corollary}

\begin{proof}
Under this assumption, the condition
$$
\liminf\limits_{t\rightarrow\infty}\dfrac{B(t)}{B_c(t,\alpha')} > 1
$$
holds for every $\alpha' \in (\sigma, \alpha)$. Now, apply Theorem~\ref{theor:second_fpt} to any such $\alpha'$ and take the limit $\alpha' \nearrow \alpha$ to conclude.
\end{proof}

In this moment, two remarks are in order.

\begin{remark}\label{rem:tto0}
Note that, on one hand, part \textit{(a)} of Theorem~\ref{theor:second_fpt} is a strict generalization of the first half of part \textit{(b)} in Theorem~\ref{mainfpt}, as follows from Corollary~\ref{cor:second_fpt}. And, on the other hand, it is also a strict generalization of Corollary~\ref{cor:mainfpt}. To see this, just take the limit $T \searrow 0$ of the upper bound in the statement of Theorem~\ref{theor:second_fpt} to find the upper bound in the statement of Corollary~\ref{cor:mainfpt} as a consequence of the weak convergence, in the sense of measures, of the heat kernel towards the Dirac delta~\cite{evanspdes}.
\end{remark}

\begin{remark}
The classification derived from the statements of Theorem~\ref{theor:second_fpt} and Corollary~\ref{cor:second_fpt} is not exhaustive. Indeed, the finiteness of the first moment of the FPT remains undecidable in certain cases, including those for which $$\liminf\limits_{t\rightarrow\infty}\dfrac{B(t)}{B_c(t,\sigma)}<1 \qquad \text{and} \qquad \limsup\limits_{t\rightarrow\infty}\dfrac{B(t)}{B_c(t,\sigma)}>1$$ hold simultaneously. Even the simpler case $$\lim\limits_{t\rightarrow\infty}\dfrac{B(t)}{B_c(t,\sigma)}=1$$
remains undecidable according to this classification (unless $B(t) \le B_c(t,\sigma)$, which is given by Theorem~\ref{mainfpt}).
\end{remark}

The following result shows that, in the decidable cases, the upper bound introduced in Theorem~\ref{theor:second_fpt} and Corollary~\ref{cor:second_fpt} can be sharpened.

\begin{lemma}\label{lem:second_fpt}
Defining $q:=\ln\left(V_0/K\right)$, under the same conditions of Corollary~\ref{cor:second_fpt}, the upper bound derived in Theorem~\ref{theor:second_fpt} can be improved to:
\begin{equation}
    \nonumber
    \mathbb{E}(\tau) \le T + 
    \int\limits_{\alpha\sqrt{T}}^\infty \frac{x^2}{ \alpha^2 - \sigma^2} \,
    \Psi \! \left(q-\mathfrak{m},x-\mathfrak{m},\sigma\sqrt{T}\right)\phi_{q,\sigma\sqrt{T}}(x) \, dx,
\end{equation}
where $\mathfrak{m} := \min \{ \tilde{B}(t), 0 \le t \le T \}$, $\tilde{B}(t):=\ln\left\{B(t) \Big/\left[Ke^{(\mu-\sigma^2/2)t}\right]\right\}$, and $\Psi(a,b,c):=1-\exp\left(-2ab/c^2\right)$, with $\phi_{a,b}(x)$ defined as in the statement of Theorem~\ref{theor:second_fpt}.

Moreover, if $\tilde{B}(t)$ is concave on $[0,T]$ (supported by the maximal chord in this interval would be enough), it holds that
$$
\mathbb{E}(\tau) \le T + 
\int\limits_{\alpha\sqrt{T}}^\infty \frac{x^2}{ \alpha^2 - \sigma^2} \,
\Psi \! \left(q,x-\alpha\sqrt{T},\sigma\sqrt{T}\right)\phi_{q,\sigma\sqrt{T}}(x) \, dx.
$$
\end{lemma}

\begin{proof}
\begin{comment}
With $\tau^\prime:= \inf \lbrace t\geq 0 \, \vert \, V_{t+T}\leq B(t+T)\rbrace$ and by the law of total expectation, we have
\begin{equation}\nonumber
\mathbb{E}(\tau) = \mathbb{E}(\tau \mathbb{I}_{\tau \leq T})+ \mathbb{E}((\tau' +T)\mathbb{I}_{\tau > T}) \le T+ \mathbb{E}(\tau' \mathbb{I}_{\tau > T}),
\end{equation}
since $\tau \mathbb{I}_{\tau \leq T} \le T \mathbb{I}_{\tau \leq T}$.
Because $B(t)$ is continuous on $[0,T]$, so is $\tilde{B}(t):=\ln\left\{B(t) \big/\left[Ke^{(\mu-\sigma^2/2)t}\right]\right\}$, and then it attains its minimum there, which we hereafter denote $\mathfrak{m}$, and which is non-positive, because $\tilde{B}(0)=0$.  We should keep in mind that
$$
\tau=\inf\lbrace t\geq 0 \, | \, V_t=B(t)\rbrace=\inf\lbrace t\geq 0 \, | \, \sigma W_t+q=\tilde{B}(t)\rbrace.$$
\end{comment}

The proof is identical to that of Theorem~\ref{theor:second_fpt} up to equation~\eqref{eq:for_lemma}, from which we depart. Instead of using the estimate $\mathbb{P}(\tau>T\vert \sigma W_T+q=y>\alpha\sqrt{T})\leq 1$ for every $y$, we will derive some tighter bounds for this conditioned survival probability. These bounds will be based on equation~(15) in~\cite{duffie2001term}, which shows that,
for an arithmetic Brownian motion $X_t=\sigma W_t+x_0+\rho t$ with $\sigma, x_0>0$, $\rho \in \mathbb{R}$, its first passage time through zero, $\tau_0:=\inf\lbrace t\geq 0 \vert X_t=0\rbrace$, fulfills
$$
\mathbb{P}\left(\tau_0>s\vert X_s=b, X_0=a\right)=\Psi\!\left(a,b,\sigma\sqrt{s}\right)1_{\lbrace a,b>0\rbrace}
$$
with $\Psi(a,b,c)$ as in the statement of this theorem and $s>0$.

For the first case, note that $\mathfrak{m}$ is well defined as $\tilde{B}(t)$ is continuous since $B(t)$ is continuous and positive. Then, define $\tau_\mathfrak{m} := \inf \{t \ge 0 \, | \, \sigma W_t+q = \mathfrak{m} \}$. By Lemma~\ref{lemma:rem2barriers}, $\tau \le \tau_\mathfrak{m}$ almost surely, and hence $\mathbb{P}(\tau > T\vert W_T,W_0) \le \mathbb{P}(\tau_\mathfrak{m} > T\vert W_T,W_0) \, 1_{V_T>B(T)}$.

\begin{comment} On the other hand, by the Markov property of $V_T$ we have:
$$
\mathbb{E}(\tau'\vert V_T,\tau > T)=\mathbb{E}(\tau'\vert V_T).
$$
Using these results, we have the estimate
\begin{eqnarray}\nonumber
\mathbb{E}(\tau' \mathbb{I}_{\tau > T}) &=& \mathbb{E}[\mathbb{E}(\tau'\mathbb{I}_{\tau > T}\vert V_T)]=\mathbb{E}[\mathbb{E}(\tau'\vert V_T,\tau > T)\mathbb{P}(\tau> T\vert V_T)]\\
\nonumber &\le&
\mathbb{E}[\mathbb{E}(\tau'\vert V_T,\tau > T)\mathbb{P}(\tau_\mathfrak{m} > T\vert V_T)\mathbb{I}_{V_T>B(T)}]
\\ \nonumber &=&
\mathbb{E}[\mathbb{E}(\tau'\vert V_T)\mathbb{P}(\tau_\mathfrak{m} > T\vert V_T)\mathbb{I}_{V_T>B(T)}].
\end{eqnarray}
In turn, using the log-normal nature of $V_T$, we have
\begin{multline}
\mathbb{E}[\mathbb{E}(\tau'\vert V_T)\mathbb{P}(\tau_\mathfrak{m} > T\vert V_T)\mathbb{I}_{V_T>B(T)}]=\;(2\pi\sigma^2T)^{-1/2}\;\times\\
\nonumber
\int\limits_{B(T)}^\infty \!\! \mathbb{E}(\tau'\vert V_T=y)\mathbb{P}(\tau_\mathfrak{m} > T\vert V_T=y)] \exp \! \left\{ -\frac{\left\vert\ln(y/V_0)-(\mu-\sigma^2/2)T \right\vert^2}{2 \sigma^2 T} \right\} \! \frac{dy}{y}.
\end{multline}
\end{comment}

Consider now the process $X_t:=\sigma W_t+q-\mathfrak{m}$ and its associated $\tau_0$ (defined as above), since for that process $\tau_0=\tau_{\mathfrak{m}}$; for $s=T$ we obtain:
$$
\mathbb{P}(\tau > T\vert \sigma W_T+q=x) \le\mathbb{P}(\tau_\mathfrak{m} > T\vert \sigma W_T+q=x)=\Psi\!\left(q-\mathfrak{m},x-\mathfrak{m},\sigma\sqrt{T}\right).
$$
Substituting this estimate in equation~\eqref{eq:for_lemma} completes the proof of the first part of this lemma.

\begin{comment}
Besides, from Corollary~\ref{cor:mainfpt}, we have
$$
\mathbb{E}(\tau'\vert V_T=y(x))\leq\dfrac{x^2}{\alpha^2-\sigma^2}.
$$
Performing the change of variables and substituting these expressions finishes the proof of the first part of this lemma.
\end{comment}

If $\tilde{B}(t)$ is concave on $[0,T]$, any of its chords within this interval lies below its graph. For the second part of this lemma, it is enough to assume that the maximal chord lies below the graph, i.e. 
$$\tilde{B}(t)\geq \dfrac{t}{T}\tilde{B}(T)=\alpha\dfrac{t}{\sqrt{T}}\;\forall \, t\in[0,T].$$
The proof of this second case goes as in the first, but instead of $\tau_{\mathfrak{m}}$ we use
\begin{eqnarray}\nonumber
\tau_\ell := \inf \! \left\{t \ge 0 : \sigma W_t +q=\dfrac{t}{T}\tilde{B}(T) \right\} =\inf \! \left\{t \ge 0 : \sigma W_t +q= \frac{\alpha}{\sqrt{T}} \, t \right\},
\end{eqnarray}
where we have used $\tilde{B}(T)=\alpha\sqrt{T}$. Once more by Lemma~\ref{lemma:rem2barriers} we know that $\tau \le \tau_\ell$ almost surely, and hence $\mathbb{P}(\tau > T\vert V_T) \le \mathbb{P}(\tau_\ell > T\vert V_T)$. This time, we consider the stochastic process $X_t=\sigma W_t+q-\alpha \, t/\sqrt{T}$, for which $\tau_0=\tau_\ell$, and thus:
$$
\mathbb{P}(\tau > T\vert \sigma W_T+q=x) \le\mathbb{P}(\tau_\ell > T\vert \sigma W_T+q=x)=\Psi \! \left(q,x-\alpha\sqrt{T},\sigma\sqrt{T} \right).
$$
Substituting this expression in equation~\eqref{eq:for_lemma} finishes the proof.
\begin{comment}
The rest of the proof goes as before. We substitute this last expression in
\begin{multline}
\mathbb{E}[\mathbb{E}(\tau'\vert V_T)\mathbb{P}(\tau_\ell > T\vert V_T)]=\;(2\pi\sigma^2T)^{-1/2}\;\times\\
\nonumber
\int\limits_{B(T)}^\infty \!\! \mathbb{E}(\tau'\vert V_T=y)\mathbb{P}(\tau_\ell> T\vert V_T=y)] \exp \! \left\{ -\frac{\left\vert\ln(y/V_0)-(\mu-\sigma^2/2)T \right\vert^2}{2 \sigma^2 T} \right\} \frac{dy}{y}.
\end{multline}
The proof concludes by applying once more the same change of variables and substituting the resulting expression in the inequality
$$
\mathbb{E}(\tau)\leq T+\mathbb{E}[\mathbb{E}(\tau'\vert V_T)\mathbb{P}(\tau_\ell > T\vert V_T)].$$
\end{comment}
\end{proof}

\begin{remark}
The first bound in this lemma is sharper than that of Corollary~\ref{cor:second_fpt}, and the second is sharper than the first (although in the second case, unlike the first one, we ask for an additional property of the barrier function, which can be thought of as a weaker form of concavity restricted to the interval $[0,T]$). This is because $\tau \le \tau_\ell \le \tau_\mathfrak{m}$ almost surely by Lemma~\ref{lemma:rem2barriers}. Also, it can be immediately seen that the limit $T \searrow 0$ reduces all three upper bounds to that in the statement of Corollary~\ref{cor:mainfpt} by Remark~\ref{rem:tto0} and the properties of the inverse Gaussian distribution~\cite{Seshadri1999}.
\end{remark}

This last technical result provides two upper bounds that are more accurate than that of Theorem~\ref{theor:second_fpt}, despite their higher complexity. All of these three bounds are explicit up to quadrature, but the integrals involved are of Gaussian type, easily computable both numerically or analytically in terms of the error function. % In the next section we will compute a further refinement of these bounds that can be obtained by placing conditions on  $\tilde{B}(t):=\ln\left\{B(t) \Big/\left[Ke^{(\mu-\sigma^2/2)t}\right]\right\}$ rather than  $B(t)$.
The next section shows how bounds on the mean FPT can be modulated when more concrete examples are considered.

\section{Examples}
\label{sec:examples}

Certain applications might suggest specific properties for the barrier functions. Throughout this section, we will illustrate how particular barriers, or families of them, can be approached with the arguments we have employed so far. This shows that complementary results to the general statements we have herein introduced are very well possible whenever additional hypotheses on the barrier functions are assumed.

We begin with an example that can be analyzed within the framework of Theorem~\ref{theor:second_fpt} but, however, we will study it with the methods used in the proof of Theorem~\ref{mainfpt}.

\begin{example}\label{examplej}
We will prove that $\mathbb{E}(\tau)<\infty$ for
$$
\tau = \inf \{t \ge 0 : V_t = B(t)\}\;\text{with}\;
B(t) = K \exp \! \left[(\mu-\sigma^2/2)t + \sigma \sqrt{t \ln(t+1)}\right] \! .
$$
This FPT can be expressed as
$$
\tau=\inf \left\{t \ge 0 : \ln(V_0/K) + \sigma W_t = \sigma \sqrt{t \ln(t+1)} \right\},
$$
which is, by Theorem~\ref{mainfpt}, finite almost surely.
As there we find
\begin{eqnarray}\nonumber
\sigma \sqrt{\tau \ln(\tau+1)} &=& \ln(V_0/K) + \sigma W_\tau, \\ \nonumber
\sigma \sqrt{(t \wedge \tau) \ln(t \wedge \tau+1)} &\le& \ln(V_0/K) + \sigma W_{t \wedge \tau},
\end{eqnarray}
so
$$
\sigma^2 \, (t \wedge \tau) \ln(t \wedge \tau+1) \le (\ln(V_0/K))^2 + \sigma^2 W_{t \wedge \tau}^2 + 2 \ln(V_0/K) \, \sigma W_{t \wedge \tau}.
$$
By taking the expectation of this inequality it follows that
\begin{eqnarray}\nonumber
\sigma^2 \, \mathbb{E} \left[(t \wedge \tau) \ln(t \wedge \tau+1) \right] &\le& (\ln(V_0/K))^2 + \sigma^2 \, \mathbb{E}(W_{t \wedge \tau}^2) \\ \nonumber &&
+ 2 \ln(V_0/K) \, \sigma \, \mathbb{E} (W_{t \wedge \tau}) 
\\
\nonumber
&=& (\ln(V_0/K))^2 + \sigma^2 \, \mathbb{E}(t \wedge \tau),
\end{eqnarray}
again by the linearity of the expectation and the optional stopping theorem along with the martingality of $W_t$ and $W_t^2 -t$. This can be rewritten as
$$
\mathbb{E} \left[(t \wedge \tau) \ln(t \wedge \tau+1) \right]
- \mathbb{E}(t \wedge \tau) \le \frac{(\ln(V_0/K))^2}{\sigma^2},
$$
and by the Jensen inequality
\begin{eqnarray}\nonumber
\mathbb{E} (t \wedge \tau) \ln \left[ \mathbb{E}(t \wedge \tau) +1 \right]
- \mathbb{E}(t \wedge \tau) &\le&
\mathbb{E} \left[(t \wedge \tau) \ln(t \wedge \tau +1) - (t \wedge \tau) \right]
\\ \nonumber
&\le& \frac{(\ln(V_0/K))^2}{\sigma^2};
\end{eqnarray}
finally, by the monotony of the logarithm, $\ln(\mathbb{E}(t \wedge \tau))\leq \ln(\mathbb{E}(t \wedge \tau)+1)$, so
\begin{eqnarray}\nonumber
\mathbb{E} (t \wedge \tau) \ln \left[ \mathbb{E}(t \wedge \tau) \right]
- \mathbb{E}(t \wedge \tau) &\le&\frac{(\ln(V_0/K))^2}{\sigma^2}.
\end{eqnarray}
This inequality can be solved to yield
$$
\mathbb{E} (t \wedge \tau) \le \frac{(\ln(V_0/K)/\sigma)^2}{\mathcal{W}_0\left[ (\ln(V_0/K)/\sigma)^2/e \right]},
$$
where $\mathcal{W}_0(\cdot)$ is the principal branch of the Lambert omega function~\cite{Corless1996LambertW}. By the monotone convergence theorem we see
$$
\mathbb{E} (\tau) = \mathbb{E} \left( \lim_{t \to \infty} t \wedge \tau \right) = \lim_{t \to \infty} \, \mathbb{E} (t \wedge \tau),$$
and thus we conclude
$$
\mathbb{E} (\tau) \le \frac{(\ln(V_0/K)/\sigma)^2}{\mathcal{W}_0\left[ (\ln(V_0/K)/\sigma)^2/e \right]};
$$
such an inequality provides an explicit upper bound for the mean FPT and consequently a proof of its finiteness. 
\end{example}

This example has provided a much simpler upper bound for the mean FPT than that given by Theorem~\ref{theor:second_fpt}. The next example shows that the same reasoning can be extended to a different class of barriers.

\begin{example}
We focus on continuous barriers of the form
$$
B(t)=K\exp\left[(\mu-\sigma^2/2)t+\tilde{B}(t)\right]
$$
with $\tilde{B}(t)$ convex and strictly increasing on $[0,\infty)$. The current assumptions imply, first, that $\tilde{B}(0)=0$ so $B(0)=K$. And, second, by Theorem~1 from Appendix~B in~\cite{evanspdes}, every convex function admits a supporting hyperplane at each of its points; since $\tilde{B}(t)$ is also strictly increasing, this means $\tilde{B}(t) \to \infty$ as $t \to \infty$ (and, moreover, the asymptotic growth must be linear or superlinear). This implies that $\tilde{B}(t)$ maps $[0,\infty)$ bijectively into itself.

The FPT $\tau:=\inf\lbrace t\geq 0 : V_t=B(t)\rbrace$ fulfills:
$$
\sigma W_\tau+q=\tilde{B}(\tau)\quad\text{and}\quad\sigma W_{t\wedge\tau}+q\geq\tilde{B}(t\wedge\tau),\;\text{ where }q:=\ln\left(\frac{V_0}{K}\right).
$$
Note that, under the stated assumptions, the barrier satisfies both conditions in  Theorem~\ref{mainfpt} (one would be enough), therefore $\tau$ is finite almost surely.
The expectation of the second equation yields
\begin{equation}\nonumber
q \ge \mathbb{E}\left[\tilde{B}(t\wedge\tau)\right] - \sigma \, \mathbb{E}\left[W_{t\wedge\tau}\right]
\ge \tilde{B}\left(\mathbb{E}[t\wedge\tau]\right),
\end{equation}
where we have used, in this order, the linearity of the expectation and the optional stopping theorem for the martingale $W_t$ along with the Jensen inequality. Then we obtain
$$
\mathbb{E}[t\wedge\tau] \le \tilde{B}^{(-1)}\left(q\right),
$$
where the inverse $\tilde{B}^{(-1)}(\cdot)$ of $\tilde{B}(\cdot)$ is well defined because the latter function is a bijection of $[0,\infty)$ into itself, $q \in \, (0,\infty)$, and the direction of the inequality is preserved since $\tilde{B}(\cdot)$ is strictly increasing. Arguing as in the previous example, we conclude
$$
\mathbb{E}(\tau) \le \tilde{B}^{(-1)}\left(q\right),
$$
what guarantees the finiteness of the mean FPT since $\tilde{B}^{(-1)}\left(\cdot\right)$ is bounded on compacts.  In fact, we already knew this because the conditions on $\tilde{B}(t)$ imply that  it fulfills the necessary conditions given by Theorem \ref{theor:second_fpt}. Examples of eligible functions are $\tilde{B}(t) \propto t^p$ with $p \ge 1$.
\end{example}

While the previous two examples have shown how upper bounds can be derived for the mean FPTs, the next will show that lower bounds are also accessible to our present methods. 

\begin{example}
Now, we change our focus to barriers of the type
$$
B(t)=K\exp\left[(\mu-\sigma^2/2)t+\tilde{B}(t)\right]
$$
with $\tilde{B}:[0,\infty)\rightarrow [0,\infty)$ continuous at the origin, strictly increasing, concave, and such that $\tilde{B}(t) \to \infty$ as $t \to \infty$. As before, from the general conditions on barriers it follows that $\tilde{B}(0)=0$; hence $\tilde{B}(t)$ is a bijection on $[0,\infty)$. Moreover, the inverse function $\tilde{B}^{(-1)}:[0,\infty)\rightarrow [0,\infty)$ exists, is continuous, and strictly increasing (the existence and continuity of the inverse of a strictly increasing and continuous function is a standard result of real analysis, see for instance Chapter 2 of \cite{royden2010real}). Now we will prove a lower bound for the mean of the FPT $\tau:=\inf\lbrace t\geq 0 : V_t=B(t)\rbrace$. First of all, as in the previous example we note that
$$
\sigma W_\tau+q=\tilde{B}(\tau) \; \text{ with }q:=\ln\left( V_0/K \right).
$$
Also, since $\tilde{B}^{(-1)}(\cdot)$ is bounded on compacts and $q>0$, we have that $\tilde{B}^{(-1)}(q) \in \, (0,\infty)$. If $\mathbb{E}(\tau)=\infty$ for a barrier of this type, the bound
$$
\mathbb{E}(\tau)\geq{\tilde{B}}^{-1}(q)
$$
holds trivially. Then, we may assume $\mathbb{E}(\tau) < \infty$. By martingality of $W_t^2-t$ and the optional stopping theorem we know that
$\mathbb{E}(W_{t\wedge\tau}^2)=\mathbb{E}(t\wedge\tau)$. Since $\mathbb{E}(\tau) < \infty$, then $\tau$ is finite almost surely and $t \wedge \tau \le \tau$ almost surely. Therefore:
$$
\sup_{t \ge 0} \mathbb{E} (W_{t\wedge\tau}^2)=\sup_{t \ge 0} \mathbb{E}(t\wedge\tau) \le \mathbb{E}(\tau).
$$
Since $t \wedge \tau \to \tau$ as $t \to \infty$ almost surely, by the Doob martingale convergence theorem we deduce that $W_{t \wedge \tau} \to W_\tau$ as $t \to \infty$ both almost surely and in $L^2(\Omega)$, and consequently in $L^1(\Omega)$, since $W_t$ is a martingale. Therefore
$$
\mathbb{E}(W_{\tau}) = \lim_{t \to \infty} \mathbb{E}(W_{t \wedge \tau}) = 0
$$
by the optional stopping theorem, and hence $\mathbb{E}(\tilde{B}(\tau))=q$ by the linearity of the expectation. Given that $\tilde{B}(t)$ is concave, $-\tilde{B}(t)$ is convex, thus, by the Jensen inequality:
$$
-\tilde{B}(E(\tau))\leq-\mathbb{E}(\tilde{B}(\tau))\Rightarrow \tilde{B}(E(\tau))\geq\mathbb{E}(\tilde{B}(\tau))=q.
$$
So, we conclude
$$
\mathbb{E}(\tau)\geq{\tilde{B}}^{-1}(q)>0
$$
again. Some examples of eligible functions are $\tilde{B}(t) \propto t^p$ with $0 < p \le 1$,  $\tilde{B}(t)\propto \sqrt{t\log(1+t)}$, and $\tilde{B}(t)\propto \sqrt{t\log(1+\log(1+t))}$. Note that these are close to the critical cases from the viewpoint of Theorem~\ref{mainfpt}.

\end{example}

Combining upper and lower bounds might yield exact results under more restrictive conditions. This is what we show next.

\begin{example}
A barrier that fulfills all the properties assumed in the two previous examples is necessarily of the form
$$
B(t)=K\exp\left[(\mu-\sigma^2/2)t+ \nu t \right]
$$
with $\nu >0$. Hence, applying the bounds in these examples we conclude that
$$
\mathbb{E}(\tau) = \ln\left( V_0/K \right)/\nu.
$$
Of course, the same could be derived from the classical results on FPTs of Brownian motion through linear barriers (see Section 3.5.C in~\cite{karatzas1991brownian}).
\end{example}

The previous examples dealt with bounds for the mean FPT. The next one illustrates how the theory introduced in this work might help to improve bounds on the survival probability obtained by means of classical arguments.

\begin{example}
Let us consider the following barrier function
$$
B(t) = \left\{ K \exp \left[(\mu-\sigma^2/2)t + \sigma \, t \, \sin(t) \right] \right\}
$$
and its associated FPT $\tau=\inf \{t \ge 0 : V_t = B(t) \} \equiv \inf \{t \ge 0 : W_t = \ln(K/V_0)/\sigma + t \, \sin(t) \}.$
Since $-t \le t \, \sin(t) \le t$, we can use the classical results for the probability of the finiteness of the FPT of Brownian motion through a linearly moving barrier (see again Section 3.5.C in~\cite{karatzas1991brownian}) to derive the crude estimate
$$
\left(\frac{K}{V_0}\right)^{2/\sigma} \le \mathbb{P}(\tau < \infty) \le 1,
$$
where $(K/V_0)^{2/\sigma} < 1$. However, one might invoke Theorem~\ref{mainfpt} to conclude that actually $\mathbb{P}(\tau < \infty) = 1$ in this case.
\end{example}

Our final example shows that a modification of the last barrier cannot be classified by a direct application of the results derived in this work. However, it can be analyzed using analogous mathematical machinery.

\begin{example}\label{undecidible}
Consider the variant of the previous barrier
$$
B(t) = \left\{ K \exp \left[(\mu-\sigma^2/2)t - \sigma \, t \, |\sin(t)| \right] \right\},
$$
then
$\tau = \inf \{t \ge 0 : V_t = B(t) \}
\equiv \inf \{t \ge 0 : W_t = \ln(K/V_0)/\sigma - t \, |\sin(t)| \}.
$
Since $-t \le t \, \sin(t) \le 0$, we can derive the same crude estimate for $\mathbb{P}(\tau < \infty)$ as before, but, in this case, our results so far cannot improve it, as this is one of the ``undecidable'' barriers from Remark~\ref{rem:undecidable_cases}. Nevertheless, one can improve this result by other means. Indeed, build the series of random variables
$$
X_n= W_{n \pi}-W_{(n-1)\pi}, \qquad n=1,2,3,\cdots\;,
$$
which are independent and identically distributed with $X_n \sim \mathcal{N}(0,\pi)$. Then, we apply the law of the iterated logarithm for series of random variables (the Hartman-Wintner law of iterated logarithm~\cite{deAcosta1983,HartmanWintner1941}), to find
$$
\liminf_{n \to \infty} \frac{\sum_{m=1}^n X_m}{\sqrt{2 \pi n \ln(\ln(n))}}=-1
$$
almost surely. Hence, since this sum is telescopic, we find:
$$
\liminf_{n \to \infty} W_{n \pi} = \liminf_{n \to \infty} \sum_{m=1}^n X_m = - \infty
$$
almost surely. And thus, for a large enough $n$ (depending on the realization of $W_t$), we have
$$
W_{n \pi} < \ln(K/V_0)/\sigma = \ln(K/V_0)/\sigma - n \pi \, |\sin(n \pi)|\;\;\text{almost surely.}
$$
Consequently, by the almost sure continuity of Brownian motion and the intermediate value theorem, we conclude that $\mathbb{P}(\tau < \infty)=1$.
\end{example}

This final example shows that the results presented in this work have been stated in quite generality, which in a sense limits their sharpness when it comes to their application to specific barrier shapes. On the other hand, the mathematical methods tend to be quite robust, and they still work for these particular shapes by changing the assumptions employed.

%\subsection{Further bounds}

%The following statement provides a refinement of the bounds for the mean FPT that were derived in the previous section. It is key to note that the former conditions on $B(t)$ are now replaced by conditions on $\tilde{B}(t)$. In particular, the weaker concavity assumption on $\tilde{B}(t)$ is natural since its critical behavior is given by the square root (see Theorem~\ref{mainfpt}).

%The derivation of upper bounds for the mean FPT throughout this work does not only provide explicit formulas, but also illustrates how these can be computed. If we assumed further properties on the barrier functions, more refined bounds would be accessible. It is plausible that, under certain modeling conditions, new hypotheses on the barrier functions needed to be introduced, resulting in more precise bounds for these more specific situations, but which derivations followed arguments akin to some of those exposed herein.

\section{Conclusions}
\label{sec:conclusions}

This work has focused on establishing a qualitative classification of financial risks based on pure mathematical criteria. Our main assumption is to consider the stock price being modeled by a geometric Brownian motion. This allows to establish relatively simple mathematical results with full precision, but at the same time highlights that the problem is non-trivial even if posed in these simplified terms. We have studied the ruin problem that consists in the first passage time of the geometric Brownian motion through a freely moving barrier, only assumed to be continuous and deterministic. Within this framework, we have established a triple classification of financial risks:
\begin{enumerate}
\item The {\it green flag zone}: the first passage time is infinite with positive probability.
\item The {\it yellow flag zone}: the first passage time is finite almost surely, but it has an infinite mean.
\item The {\it red flag zone}: the mean first passage time is finite (and therefore the first passage time is finite almost surely).
\end{enumerate}
Although a practical application of this classification would need a reformulation in each particular case, it shows two things: How such a classification can be carried out on a strictly mathematical basis and the difficulty in making this classification exhaustive even in the present idealized terms.

\begin{comment}
Our results are able to determine the almost sure finiteness of the FPT of a (geometric) Brownian motion through a great variety of generic barriers according to their limit behavior in relation with the law of the iterated logarithm. In case of a.s. finiteness of the FPT, we are also able to determine the finiteness of its mean depending on the limit behavior of the barrier. In this line, we also
provide explicit analytical expressions for the mean of  the FPT of the geometric Brownian motion process, $V_t=V_0\exp((\mu-\sigma^2/2)t+\sigma W_t)$ through the barrier $B_c(t,\alpha)=K\exp((\mu-\sigma^2/2)t+\alpha\sqrt{t})$. This is the same as the FPT of the (shifted and scaled) Brownian motion $\sigma W_t+q$ (with $\sigma>0$ and $q=\ln(V_0/K)>0$) through barriers of the shape $\tilde{B}_c(t)=\alpha\sqrt{t}$. These results are not the same as the classical ones on square root barriers that get further away from the expected value of the Brownian process (see e.g. \cite{patie2004some}), as ours gets closer instead.
Besides, we provide upper bounds for the mean FPT of any barrier such that $\liminf\limits_{t\rightarrow\infty}\frac{B(t)}{B_c(t,\alpha)}>1$ with $\alpha>\sigma$. Contrarily,  if $\limsup\limits_{t\rightarrow\infty}\frac{B(t)}{B_c(t,\sigma)}<1$ then its mean FPT is infinite.
\end{comment}

To be precise, we have considered our geometric Brownian motion to be
$$
V_t=V_0 \, \exp \left[(\mu-\sigma^2/2)t+\sigma W_t \right],
$$
where $W_t$ is a standard Brownian motion. We have analyzed the first passage problem of this process through a deterministic barrier $B(t)$, which is assumed to be a continuous function (initialized at $B(0) \equiv K$). Continuous functions are known to be potentially very pathological, as the literature on singular functions highlights, see for instance~\cite{Falconer2004,FernandezSanchez2012,Kawamura2011,ParadisViaderBibiloni2001}. Despite of this possible irregularity, we have found that the finiteness of the first passage time and its mean can be classified only according to the long-time asymptotic behavior of this function. To partially summarize our results let us define:
\begin{eqnarray}\nonumber
I_\pm &:=& \liminf_{t \to \infty} B(t) \Big/ \left\{ K \exp \left[(\mu-\sigma^2/2)t \pm \sigma \sqrt{2t \ln(\ln(t))}\right] \right\},
\\ \nonumber
S_\pm &:=& \limsup_{t \to \infty} B(t) \Big/ \left\{ K \exp \left[(\mu-\sigma^2/2)t \pm \sigma \sqrt{2t \ln(\ln(t))}\right] \right\}.
\end{eqnarray}
Note that $S_- \ge S_+ \ge I_+$ and $S_- \ge I_- \ge I_+$. We also define
\begin{eqnarray}\nonumber
\bar{I}_\varepsilon &:=& \liminf_{t \to \infty} B(t) \Big/ \left\{ K \exp \left[(\mu-\sigma^2/2)t + (\sigma+\epsilon)\sqrt{t}\right] \right\},
\\ \nonumber
\bar{S}_0 &:=& \limsup_{t \to \infty} B(t) \Big/ \left\{ K \exp \left[(\mu-\sigma^2/2)t + \sigma \sqrt{t}\right] \right\},
\end{eqnarray}
where $\varepsilon \ge 0$. Note also that, uniformly in $\varepsilon \ge 0$, $\bar{S}_0 \ge \bar{I}_\varepsilon$, $I_- \ge 
\bar{I}_\varepsilon \ge I_+$, and $S_- \ge 
\bar{S}_0 \ge S_+$, so consequently $\bar{S}_0 \ge I_+$ and $S_- \ge \bar{I}_\varepsilon$; obviously, the order relation among these six quantities is only partial. They allow us to summarize part of our results as follows:

\begin{itemize}
\item Case $\bar{I}_\varepsilon \ge 1$ for some $\varepsilon>0$, no matter how small. By Theorem~\ref{theor:second_fpt} and Corollary~\ref{cor:second_fpt}, $\mathbb{E}(\tau)<\infty$ and hence $\tau$ is finite almost surely, i.e. $\mathbb{P}(\tau<\infty)=1$. So this case lies in the red flag zone. Moreover, explicit upper bounds for $\mathbb{E}(\tau)$ have been provided in Lemma~\ref{lem:second_fpt}.

\item Case $I_- \ge 1$ and simultaneously $\bar{S}_0 < 1$. By Theorems~\ref{theor:maincb} and~\ref{theor:second_fpt}, $\tau$ is finite almost surely, while at the same time $\mathbb{E}(\tau)=\infty$. Then, this case lies in the yellow flag zone.

\item Case $S_- < 1$. By Theorem~\ref{theor:tlti}, $\mathbb{P}(\tau = \infty)>0$, implying $\mathbb{E}(\tau)=\infty$. Thus, this case lies in the green flag zone.

\item Although this classification is rather extensive, it is not exhaustive in the sense that some barriers cannot be classified within a colored flag zone:

\begin{itemize}
\item For instance, by Remark~\ref{rem:undecidable_cases}, the case for which $S_- \ge 1$, $I_- < 1$, and $S_+ \le 1$ remains undecidable in the sense that we cannot establish whether $\mathbb{P}(\tau = \infty)>0$ or $\mathbb{P}(\tau = \infty)=0$ for this barrier. Yet, barriers of this type can be analyzed by other methods, as Example~\ref{undecidible} shows.

\item We can even distinguish between a {\it twilight zone} for which we know that $\mathbb{E}(\tau)=\infty$ but do not know about the almost sure finiteness of $\tau$, and a {\it dark zone} for which we know nothing about the finiteness of any of them. The first contains the case $S_- \ge 1$, $I_- < 1$, and $\bar{S}_0 < 1$ (by Theorem~\ref{theor:second_fpt}), while the second includes the case $\bar{S}_0 \ge 1$, $I_- < 1$, and $S_+ \le 1$ (except when the barrier coincides with the critical one).

\item The twilight zone with $\mathbb{P}(\tau<\infty)=1$ but no knowledge on the finiteness of the mean of $\tau$ is also possible and realized in cases like $S_+>1$ and simultaneously $\bar{I}_0 < 1$ (by Theorems~\ref{mainfpt} and~\ref{theor:second_fpt}). Clearly, no other types of twilight zones are possible.
\end{itemize}
\end{itemize}

It is interesting to note that this classification does not respect any clear symmetry. For example, this can be seen in the absence of the quantity $I_+$ within it. Perhaps, a way to improve this classification is to exploit a hidden symmetry or to include this absent quantity. In any case, these are just two possibilities among many potential alternatives.

This work can be, in fact, generalized along many different directions. One of them would be to consider time-dependent parameters for the Black-Scholes model. The case of time-dependent drift $\mu(t)$ but constant volatility $\sigma$ can be addressed with the present results by introducing the transformation
$$
B(t)=K \exp \!\left(\int\limits_0^t \mu(s) \, ds -\sigma^2 t/2+\tilde{B}(t) \!\right);
$$
this reduces the problem again to analyze the FPT of the process $\ln(V_0/K)+\sigma W_t$ through the barrier $\tilde{B}(t)$, as we have already done herein. Diving deeper into this direction might be possible through the introduction of a piecewise constant volatility, or through interdependent parameters, such that $\mu(t)\propto\sigma(t)$, seizing the methods in~\cite{brigo2009credit,molini2011first}. It would also be interesting to consider barriers that are discontinuous and/or stochastic, like for example c\`adl\`ag processes. On the modeling side, our results could also be extended by considering limited access to the market, such as possessing incomplete and/or noisy information about the process $V_t$, like in~\cite{duffie2001term}. This might be specially interesting in regards to the applications, as it is known to help linking structural and intensity models~\cite{jarrow2012structural}, and to reproduce the short-maturity-time non-null credit spread seen in real markets~\cite{duffie2001term} (that many structural models fail to do~\cite{goldstein2001ebit}). Abounding in this direction, the introduction of specific models for barrier functions, particularly those with practical meaning, would contribute to a better characterization of their asymptotic behavior and consequently to hone our classification. On the theoretical side, there are some methods that we have not used so far, but that can be useful in sharpening or complementing our developments. From our viewpoint, these include the Strassen law of the iterated logarithm~\cite{Strassen1964,SzuszVolkmann1982} and the theories of large~\cite{GerholdGerstenecker2020}, moderate~\cite{Chen2005}, and small deviations~\cite{deAcosta1983smalldeviations}.

\begin{comment}
Regarding exclusively the mathematical advancements, our analysis is general and the almost sure finiteness of the FPT and the finiteness of its mean can be determined for most barrier shapes. However, for barriers such that $\limsup\limits_{t\rightarrow\infty}B(t)\neq \liminf\limits_{t\rightarrow\infty}B(t)$ we may still find some undecidable cases for the a.s. finiteness of both $\tau$ and $\mathbb{E}(\tau)$. The most immediate generalization of the current work is to include further conditions on the limit behavior of the barriers, honing the characterization provided so far and reducing the number of cases that are undecidable, and characterizing this criticality related to the limit behavior of the barriers for all the higher-order moments of the FPT.
\end{comment}

In summary, we have mathematically characterized three zones of high, medium, and low risk of default in a Black-Scholes market. While our work is purely methodological, our inspiration comes from the discrete classification of risk zones made by credit rating agencies like S\&P Global Ratings, Moody's, and Fitch Group. Evidently, their classifications are far more complex, whilst ours relies only on determining the finiteness of the first passage time of the value process through an arbitrarily moving barrier, as well as the finiteness of its mean. Of course, the advantage of our set of simplified assumptions is that it allows to prove precise mathematical statements
without compromising the freedom of choice of arbitrary continuous barriers. Still, when the asymptotic behavior of the barrier is oscillatory, there appear difficulties in building simple criteria to classify its risk (when the behavior is not oscillatory, the situation becomes simpler, see Remark~\ref{rem:undecidable_cases}). This result might be a call of attention to potential attempts of risk classification in changing environments, like those caused by extreme weather conditions. Our mathematical framework is suitable for financial scenarios with infinite maturity times, which could be regarded as a model for those cases in which debt is continuously refinanced (such as sovereign debt). However, in other frequent financial scenarios, it is relevant to know if default occurs before a finite time known as the maturity of debt; if this time cannot be considered as a distant future in a given situation, then our results would not be directly applicable (in such a case, perhaps, one can borrow more inspiration from the physics of phase transitions to include in our theory the so-called {\it finite-size effects}~\cite{quantum2021zinn}). All in all, our analysis refines the classical conclusions drawn from the Black-Cox model, as it focuses on a more detailed time structure of the debt dynamics. In practical terms, it might serve as a first step to ascertain how safety covenants must be modulated in cases considered doubtful by traditional models, or what is the effect of a non-steady inflationary rate, as explained in Section~\ref{sec:Inflation}.

\newpage

\section*{Acknowledgements}
    This work has been partially supported by the Government of Spain (Ministerio de Ciencia e Innovaci\'on) and the European Union through Projects PID2021 - 125871NB - I00, TED2021 - 131844B - I00 / AEI / 10.13039 / 501100011033 / Uni\'on Europea Next Generation EU / PRTR, and CPP2021 - 008644 / AEI / 10.13039 / 501100011033 / Uni\'on Europea Next Generation EU / PRTR.

\bibliography{main}
\bibliographystyle{plain}
\end{document}